\documentclass[acmsmall,nonacm]{acmart} 
\acmJournal{TOCL}
\setcopyright{none}
\theoremstyle{acmplain}
\newtheorem{theorem}{Theorem}[section]
\newtheorem{fact}[theorem]{Fact}

\theoremstyle{acmdefinition}
\newtheorem{alg}[theorem]{Algorithm} 
 \newtheorem{rem}[theorem]{Remark} 

\usepackage{microtype}
\usepackage{amsmath, amsfonts, 
  stmaryrd, bbm, bussproofs,xspace}
\usepackage[matrix,arrow,curve]{xy}
\usepackage{textcomp}
\usepackage[author=anonymous,nomargin,marginclue,footnote,status=draft]{fixme}
\FXRegisterAuthor{ls}{als}{LS}
\FXRegisterAuthor{dp}{adp}{DP}
 \usepackage[noend]{algpseudocode}

 \usepackage{enumitem}
 \newcommand{\children}[1]{\mathsf{ch}(#1)} \newcommand{\Gf}{{G_f}}
 \newcommand{\Sigmavars}{V_\Sigma}
 \newcommand{\Sigmasubst}{\sigma_\Sigma} \newcommand{\nneg}{{\sim}}

\newcommand{\by}[1]{(\text{#1})}
\newcommand{\gldiamond}[1]{\Diamond_{#1}}
\newcommand{\univbox}{\mathop{[\forall]}}
\newcommand{\inhom}{v}
\newcommand{\At}{\mathsf A}
\newcommand{\BC}{\mathbf C}
\newcommand{\Roles}{\mathsf R}
\newcommand{\otto}{\leftrightarrow}
\newcommand{\Rules}{\mathcal{R}}

\newcommand{\Noms}{\mathsf{N}}

\newcommand{\Land}{\bigwedge}

\newcommand{\infrule}[2]{\frac{#1}{#2}}
\newcommand{\States}{\mathsf{States}}
\newcommand{\Nodes}{G}
\newcommand{\Seqs}{\mathsf{Seqs}}
\newcommand{\type}{\Gamma}
\newcommand{\typeb}{\Delta}
\newcommand{\types}[1]{\CT(#1)}
\newcommand{\CT}{\mathcal{T}}
\newcommand{\CO}{\mathcal{O}}

\newcommand{\CE}{\mathcal{E}}
\newcommand{\CA}{\mathcal{A}}
\newcommand{\CM}{{\mathcal{M}}}

\newcommand{\Rels}{\mathsf{Rels}}
\newcommand{\Prop}{\mathsf{Prop}}
\newcommand{\NP}{{\upshape\textsc{NP}}\xspace}
\newcommand{\PSpace}{{\upshape\textsc{PSpace}}\xspace}

\newcommand{\ExpTime}{{\mbox{\upshape\textsc{ExpTime}}}\xspace}

\newcommand{\Nat}{\mathbb{N}}
\newcommand{\Rat}{\mathbb{Q}}

\newcommand{\Var}{\mathsf{Var}}

\newcommand{\hearts}{\heartsuit}

\newcommand{\ZZ}{\mathbb{Z}}

\newcommand{\Pow}{\mathcal{P}}
\newcommand{\Bag}{\mathcal{B}}
\newcommand{\Dist}{\mathcal{D}}
\newcommand{\SDist}{\mathcal{S}}

\newcommand{\PV}{\mathcal{V}}

\newcommand{\lsem}{\llbracket}
\newcommand{\rsem}{\rrbracket}

\renewcommand{\theta}{\vartheta}

\newcommand{\PLentails}{\vdash_{\mathit{PL}}}

\newcommand{\Set}{\mathsf{Set}}
\newcommand{\FLang}{\mathcal{F}}
\newcommand{\Sem}[1]{\lsem #1 \rsem}
\newcommand{\Op}{\mathit{op}}
\newcommand{\ms}[1]{\mathsf{#1}}

\newcommand{\ALC}{\mathcal{ALC}}
\newcommand{\ALCO}{\mathcal{ALCO}}
\newcommand{\ALCOQ}{\mathcal{ALCOQ}}

\newcommand{\ALCN}{\mathcal{ALCN}}
\renewcommand{\land}{\wedge}
\renewcommand{\lor}{\vee}

\newcommand{\rcount}[2]{\sharp_#1 #2}
\newcommand{\eqmod}[1]{\equiv_{#1}}

\makeatletter
\newcommand*{\@old@slash}{}\let\@old@slash\slash
\def\slash{\relax\ifmmode\delimiter"502F30E\mathopen{}\else\@old@slash\fi}
\makeatother

\newlength{\myboxwidth}
	{\end{center}\end{figure*}}

\newcounter{blubber}

\newenvironment{myitemize}
{\begin{itemize}
\setlength{\itemsep}{0.1ex}
\setlength{\parsep}{0cm}
}
{\end{itemize}}

\begin{document}

\title[Coalgebraic Reasoning with Global
  Assumptions in Arithmetic Modal Logics]{Coalgebraic Reasoning with Global
  Assumptions in\\ Arithmetic Modal Logics}
\author{Clemens
  Kupke}
\affiliation{\institution{University of Strathclyde}
  \city{Glasgow}
  \country{UK}}
\author{Dirk Pattinson}
\affiliation{\institution{Australian National University}
  \city{Canberra}
  \country{Australia}}
\author{Lutz Schr\"oder}
\affiliation{
  \institution{Friedrich-Alexander Universit\"{a}t Erlangen-N\"{u}rnberg}
\city{Erlangen}\country{Germany}}

\begin{abstract} We establish a generic upper bound \ExpTime for reasoning with global assumptions (also known as TBoxes) in coalgebraic modal logics. Unlike earlier results of this kind, our bound does not require a tractable set of tableau rules for the instance logics, so that the result applies to wider classes of logics. Examples are Presburger modal logic, which extends graded modal logic with linear inequalities over numbers of successors, and probabilistic modal logic with polynomial inequalities over probabilities. We establish the theoretical upper bound using a type elimination algorithm. We also provide a global caching algorithm that potentially avoids building the entire exponential-sized space of candidate states, and thus offers a basis for practical reasoning. This algorithm still involves frequent fixpoint computations; we show how these can be handled efficiently in a concrete algorithm modelled on Liu and Smolka's linear-time fixpoint algorithm. Finally, we show that the upper complexity bound is preserved under adding nominals to the logic, i.e.\ in coalgebraic hybrid logic.
\end{abstract}

\maketitle

\section{Introduction}\label{sec:intro}

\noindent While modal logic is classically concerned with purely
relational systems (e.g.~\cite{BlackburnEA01}), there is, nowadays,
widespread interest in flavours of modal logic interpreted over
state-based structures in a wider sense, e.g.\ featuring probabilistic
or, more generally, weighted branching. Under the term
\emph{arithmetic modal logics}, we subsume logics that feature
arithmetical constraints on the number or combined weight of
successors. The simplest logics of this type compare weights to
constants, such as graded modal logic~\cite{Fine72} or some variants
of probabilistic modal logic~\cite{LarsenSkou91,HeifetzMongin01}. More
involved examples are \emph{Presburger modal
  logic}~\cite{DemriLugiez10}, which allows Presburger constraints on
numbers of successors, and probabilistic modal logic with
linear~\cite{FaginHalpern94} or
polynomial~\cite{FaginHalpernMegiddo90} inequalities over
probabilities. Presburger modal logic allows for statements like `the
majority of university students are female', or `dance classes have
even numbers of participants', while probabilistic modal logic with
polynomial inequalities can assert, for example, independence of
events.

These logics are the main examples we address in a more general
coalgebraic framework in this paper. Our main observation is that
satisfiability for coalgebraic logics can be decided in a step-by-step
fashion, peeling off one layer of operators at a time. We thus reduce
the overall satisfiability problem to satisfiability in a
\emph{one-step logic} involving only immediate successor states, and
hence no nesting of
modalities~\cite{SchroderPattinson08d,MyersEA09}. We define a
\emph{strict} variant of this \emph{one-step satisfiability problem},
distinguished by a judicious redefinition of its input size; if strict
one-step satisfiability is in \ExpTime, we obtain a (typically
optimal) \ExpTime upper bound for satisfiability under global
assumptions in the full logic. For our two main examples, the
requisite complexity bounds (in fact, even \PSpace) on strict one-step
satisfiability follow in essence directly from known complexity
results in integer programming and the existential theory of the
reals, respectively; in other words, even in fairly involved examples
the complexity bound for the full logic is obtained with comparatively
little effort once the generic result is in place.

Applied to Presburger constraints, our results complement previous
work showing that the complexity of Presburger modal logic without
global assumptions is \PSpace~\cite{DemriLugiez06,DemriLugiez10}, the
same as for the modal logic $K$ (or equivalently the description logic
$\ALC$).  For polynomial inequalities on probabilities, our syntax
generalizes propositional \emph{polynomial weight}
formulae~\cite{FaginHalpernMegiddo90} to a full modal logic allowing
nesting of weights (and global assumptions).

In more detail, our first contribution is to show via a type
elimination algorithm~\cite{Pratt79} that also in presence of global
assumptions (and, hence, in presence of the universal
modality~\cite{GorankoPassy92}),
the satisfiability problem for coalgebraic modal logics is no harder
than for $K$, i.e.\ in \ExpTime, provided that strict one-step
satisfiability is in \ExpTime. Additionally, we show that this result
can be extended to cover nominals, i.e.\ to coalgebraic hybrid
logic~\cite{MyersEA09,SchroderEA09}.  In the Presburger example, we
thus obtain that reasoning with global assumptions in Presburger
hybrid logic, equivalently reasoning with general TBoxes in the
extension of the description logic $\ALCO$ with Presburger constraints
(which subsumes $\ALCOQ$), remains in \ExpTime.

We subsequently refine the algorithm to use global caching in the
spirit of Gor\'e and Nguyen~\cite{GoreNguyen13}, i.e.\ bottom-up
expansion of a tableau-like graph and propagation of satisfiability
and unsatisfiability through the graph. We thus potentially avoid
constructing the whole exponential-sized tableau, and provide
maneuvering space for heuristic optimization. Global caching
algorithms have been demonstrated to perform well in
practice~\cite{Gore:2008:EEG}. Moreover, we go on to present a
concrete algorithm, in which the fixpoint computations featuring in
the propagation step of the global caching algorithm are implemented
efficiently in the style of Liu and Smolka~\cite{lism98:simp}.

\paragraph{Organization} We discuss some preliminaries on fixpoints in
Section~\ref{sec:prelims}, and recall the generic framework of
coalgebraic logic in Section~\ref{sec:colog}. In
Section~\ref{sec:oss}, we discuss the concepts of one-step logic and
one-step satisfiability that underlie our generic algorithms. We
establish the generic \ExpTime upper bound for reasoning with global
assumptions in coalgebraic modal logics via type elimination in
Section~\ref{sec:type-elim}. In Sections~\ref{sec:caching}
and~\ref{sec:concrete-alg}, we present the global caching algorithm
and its concretization. We extend the \ExpTime complexity result to
coalgebraic hybrid logics in Section~\ref{sec:nominals}.

\paragraph{Related Work} Our algorithms use a semantic method, and as
such complement earlier results on global caching in coalgebraic
description logics that rely on tractable sets of tableau
rules~\cite{GoreEA10a}, which are not currently available for our
leading examples.  (In fact, tableau-style axiomatizations of various
logics of linear inequalities over the reals and over the integers
have been given in earlier work~\cite{Kupke:2010:MLL}; however, over
the integers the rules appear to be incomplete: if $\sharp p$ denotes
the integer weight of successors satisfying $p$, then the formula
$2\sharp\top <1\lor 2\sharp\top>1$ is clearly valid, but cannot be
derived.)

Demri and Lugiez' proof that Presburger modal logic \emph{without}
global assumptions is in \PSpace~\cite{DemriLugiez06,DemriLugiez10}
can be viewed as showing that strict one-step satisfiability in
Presburger modal logic is in \PSpace (as we discuss below, more recent
results in integer programming simplify this proof). Generally, our
coalgebraic treatment of Presburger modal logic and related logics
relies on an equivalence of the standard Kripke semantics of these
logics and an alternative semantics in terms of non-negative-integer-weighted
systems called \emph{multigraphs}~\cite{DAgostinoVisser02}, the point
being that the latter, unlike the former, is subsumed by the semantic
framework of coalgebraic logic (we explain details in
Section~\ref{sec:colog}).

Work related to XML query languages has shown that reasoning in
Presburger fixpoint logic is \ExpTime complete~\cite{SeidlEA08}, and
that a logic with Presburger constraints and nominals is in
\ExpTime~\cite{BarcenasLavalle13}, when these logics are interpreted
\emph{over finite trees}, thus not subsuming our \ExpTime upper bound
for Presburger modal logic with global assumptions. It may be possible
to obtain the latter bound alternatively via looping tree automata
like for graded modal logic~\cite{TobiesThesis}.  The description
logic~$\ALCN$ (featuring the basic~$\ALC$ operators and number
restrictions $\ge n.\,\top$) has been extended with explicit
quantification over integer variables and number restrictions
mentioning integer variables~\cite{BaaderSattler96}, in formulae such
as $\downarrow n.\,(({=} n\ R.\,\top)\land ({=}n\ S.\,\top))$ with $n$
an integer variable, and $\downarrow$ read as existential
quantification, so the example formula says that there are as many
$R$-successors as $S$-successors. This logic remains decidable if
quantification is restricted to be existential. It appears to be
incomparable to Presburger modal logic in that it does not support
general linear inequalities or qualified number restrictions, but on
the other hand allows the same integer variable to be used at
different modal depths.

Reasoning with polynomial inequalities over probabilities has been
studied in propositional logics~\cite{FaginHalpernMegiddo90} and in
many-dimensional modal logics~\cite{GutierrezBasultoEA17}, which work
with a single distribution on worlds rather than with world-dependent
probability distributions as
in~\cite{LarsenSkou91,HeifetzMongin01,FaginHalpern94}.

This paper is a revised and extended version of a previous conference
publication~\cite{KupkeEA15}; besides including full proofs and
additional examples, it contains new material on the concretized
version of the global caching algorithm
(Section~\ref{sec:concrete-alg}) and on \ExpTime reasoning with global
assumptions in coalgebraic hybrid logics (Section~\ref{sec:nominals}).

\section{Preliminaries}\label{sec:prelims}

\noindent Our reasoning algorithms will centrally involve fixpoint
computations on powersets of finite sets; we recall some
notation. Let~$X$ be a finite set, and let $F\colon\Pow X\to\Pow X$ be
a function that is monotone with respect to set inclusion. A set
$Y\in\Pow X$ is a \emph{prefixpoint} of~$F$ if $F(Y)\subseteq Y$; a
\emph{postfixpoint} of~$F$ if $Y\subseteq F(Y)$; and a \emph{fixpoint}
of~$F$ if $Y=F(Y)$. By the Knaster-Tarski fixpoint theorem,~$F$ has a
least fixpoint $\mu F$ and a greatest fixpoint $\nu F$. Moreover,
$\mu F$ is even the least prefixpoint of~$F$, and $\nu F$ the greatest
postfixpoint. We alternatively use a $\mu$-calculus-like notation,
writing $\mu S.\,E(S)$ and $\nu S.\,E(S)$ for the least and greatest
fixpoints, respectively, of the function on~$\Pow X$ that maps
$S\in\Pow X$ to $E(S)$, where~$E$ is an expression (in an informal
sense) depending on~$S$. Since~$X$ is finite, we can compute least and
greatest fixpoints by \emph{fixpoint iteration} according to Kleene's
fixpoint theorem: Given a monotone~$F$ as above, the sets
$F^n(\emptyset)$ (where $F^n$ denotes $n$-fold application of~$F$)
form an ascending chain
\begin{equation*}
  \emptyset = F^0(\emptyset) \subseteq F(\emptyset) \subseteq
  F^2(\emptyset)\subseteq\dots,
\end{equation*}
which must stabilize at some $F^k(\emptyset)$ (i.e.\
$F^{k+1}(\emptyset)=F^k(\emptyset)$), and then $\mu
F=F^k(\emptyset)$. Similarly, the sets $F^n(X)$ form a descending
chain, which must stabilize at some $F^k(X)$, and then $\nu F=F^k(X)$.

\section{Coalgebraic Logic}\label{sec:colog}

As indicated above, we cast our results in the generic framework of
\emph{coalgebraic logic}~\cite{CirsteaEA11}, which allows us to treat
structurally different modal logics, such as Presburger and
probabilistic modal logics, in a uniform way. We briefly recall the
main concepts needed. Familiarity with basic concepts of category
theory (e.g.~\cite{Awodey10}) will be helpful, but we will explain the
requisite definitions as far as necessary for the present
purposes. Overall, coalgebraic logic is concerned with the
specification of state-based systems in a general sense by means of
\emph{modalities}, which are logical connectives that traverse the
transition structure in specific ways. The basic example of such a
modal logic is what for our present purposes we shall term
\emph{relational} modal logic (e.g.~\cite{BlackburnEA01}). Here,
states are connected by a successor relation, and modalities
$\Box,\Diamond$ talk about the successors of a state: a formula of the
form~$\Box\phi$ holds for a state if \emph{all} its successors
satisfy~$\phi$, and a formula of the form~$\Diamond\phi$ holds for a
state if it has \emph{some} successor that satisfies~$\phi$. Our main
interest, however, is in logics where the transition structure of
states goes beyond a simple successor relation, with correspondingly
adapted, and often more complex, modalities.

We parametrize modal logics in terms of their syntax and their
coalgebraic semantics. In the \textbf{syntax}, we work with a modal
similarity type $\Lambda$ of modal operators with given finite
arities.  The set $\FLang(\Lambda)$ of \emph{$\Lambda$-formulae} is
then given by the grammar
\begin{equation}\label{eq:grammar}
  \FLang(\Lambda)\owns\phi,\psi::= \bot\mid\phi\land\psi\mid
  \neg\phi\mid \hearts(\phi_1,\dots,\phi_n)\qquad (\hearts\in\Lambda\text{ $n$-ary}).
\end{equation}
We omit explicit propositional atoms; these can be regarded as nullary
modalities.  The operators $\top$, $\to$, $\lor$, $\otto$ are assumed
to be defined in the standard way. Standard examples of modal
operators include the mentioned (unary) box and diamond operators
$\Box,\Diamond$ of relational modal logic; as indicated above, in the
present setting, our main interest is in more complex examples
introduced in Sections~\ref{sec:presburger} and~\ref{sec:prob}. For
the complexity analysis of reasoning problems, we assume a suitable
encoding of the modal operators in~$\Lambda$ as strings over some
alphabet. The \emph{size}~$|\phi|$ of a formula~$\phi$ is then defined
by counting~$1$ for each Boolean operation ($\bot$, $\neg$, $\land$),
and for each modality~$\hearts\in\Lambda$ the length of the encoding
of~$\hearts$.  We assume that numbers occurring in the description of
modal operators are coded in binary. To ease notation, we generally
let $\epsilon\phi$, for $\epsilon\in\{-1,1\}$, denote $\phi$ if
$\epsilon=1$ and $\neg\phi$ if $\epsilon=-1$.

The \textbf{semantics} of the logic is formulated in the paradigm of
\emph{universal coalgebra}~\cite{Rutten00}, in which a wide range of
state-based system types, e.g.\ relational, neighbourhood-based,
probabilistic, weighted, or game-based systems, is subsumed under the
notion of functor coalgebra. Here, a \emph{functor}~$T$ on the
category of sets assigns to each set~$X$ a set~$TX$, thought of as a
type of structured collections over~$X$, and to each map
$f\colon X\to Y$ a map $Tf\colon TX\to TX$, preserving identities and
composition. A standard example is the \emph{(covariant) powerset
  functor} $\Pow$, which maps a set~$X$ to its powerset $\Pow X$ and a
map $f\colon X\to Y$ to the direct image map
$\Pow f\colon \Pow X\to\Pow Y$, i.e.\ $(\Pow f)(A)=f[A]$ for
$A\in\Pow X$. In this case, structured collections are thus just
sets. A further example, more relevant to our present purposes, and to
be taken up again in Section~\ref{sec:prob}, is the \emph{(discrete)
  distribution functor} $\Dist$. This functor assigns to a set~$X$ the
set of discrete probability distributions on~$X$, which thus play the
role of structured collections, and to a map $f\colon X\to Y$ the map
$\Dist f\colon\Dist X\to\Dist Y$ that takes image measures; i.e.\
$(\Dist f)(\mu)(B)=\mu(f^{-1}[B])$ for $B\subseteq Y$. We recall here
that a probability distribution~$\mu$ on~$X$ is \emph{discrete} if
$\mu(A)=\sum_{x\in A}\mu(\{x\})$ for every $A\subseteq X$, i.e.\ we
can equivalently regard~$\mu$ as being given by its \emph{probability
  mass function} $x\mapsto\mu(\{x\})$. Note that the support
$\{x\mid \mu(\{x\})\neq 0\}$ of~$\mu$ is then necessarily countable. A
functor~$T$ defines a system type in the shape of its class of
\emph{$T$-coalgebras}, which are pairs $C=(X,\gamma)$ consisting of a
set~$X$ of \emph{states} and a \emph{transition map}
\begin{equation}\label{eq:transition}
  \gamma\colon X\to TX,
\end{equation}
thought of as assigning to each state~$x$ a structured
collection~$\gamma(x)\in TX$ of successors. For instance,
$\Pow$-coalgebras are just transition systems or Kripke frames, as
they assign to each state a set of successors (i.e.~they capture
precisely the semantic structures that underlie relational modal logic
as recalled at the beginning of the section), and $\Dist$-coalgebras
are Markov chains, as they assign to each state a distribution over
successors.

We further parametrize the semantics over an interpretation of
modalities as predicate liftings, as follows.
Recall~\cite{Pattinson04,Schroder08} that an \emph{$n$-ary predicate
  lifting} for $T$ is a natural transformation
\begin{equation*}
  \lambda\colon Q^n\to Q\circ T^\Op
\end{equation*}
where~$Q$ denotes the \emph{contravariant powerset functor}. We shall
use predicate liftings in connection with the transition
map~\eqref{eq:transition} to let modalities look one step ahead in the
transition structure of a coalgebra. The definition of predicate
liftings unfolds as follows. Recall that every category~$\BC$ has a
\emph{dual} category~$\BC^\Op$, which has the same objects as~$\BC$
and the same morphisms, but with the direction of morphisms
reversed. In particular, $\Set^\Op$, the dual category of the
category~$\Set$ of sets and maps, has sets as objects, and maps
$Y\to X$ as morphisms $X\to Y$. Then the contravariant powerset
functor $Q\colon \Set^\Op\to\Set$ assigns to a set~$X$ its powerset
$QX=\Pow X$, and to a map $f:X\to Y$ the preimage map $Q f:Q Y\to QX$,
given by $(Qf)(B)=f^{-1}[B]$ for $B\subseteq Y$. By $Q^n$, we denote
the pointwise $n$-th Cartesian power of~$Q$, i.e.\ $Q^nX=(QX)^n$. The
functor $T^\Op\colon \Set^\Op\to\Set^\Op$ acts like~$T$. Thus,
$\lambda$ is a family of maps $\lambda_X\colon (QX)^n\to Q(TX)$
indexed over all sets~$X$, satisfying the \emph{naturality} equation
$\lambda_X\circ (Q f)^{n}=Q(T^\Op f)\circ\lambda_Y$ for
$f\colon X\to Y$. That is,~$\lambda_X$ takes~$n$ subsets of~$X$ as
arguments, and returns a subset of~$TX$. The naturality condition
amounts to commutation of~$\lambda$ with preimage, i.e.\
\begin{equation}\label{eq:naturality}
\lambda_X(f^{-1}[B_1],\dots,f^{-1}[B_n])=Tf^{-1}[\lambda_Y(B_1,\dots,B_n)]
\end{equation}
for $B_1,\dots,B_n\subseteq Y$. We assign an $n$-ary predicate lifting
$\Sem{\hearts}$ to each modality $\hearts\in\Lambda$, of arity~$n$,
thus determining the semantics of~$\hearts$.  For $t\in TX$ and
$A_1,\dots,A_n\subseteq TX$, we write
\begin{equation}\label{eq:lifting-notation}
t\models\hearts(A_1,\dots,A_n)
\end{equation}
to abbreviate $t\in\Sem{\hearts}_X(A_1,\dots,A_n)$.

Predicate liftings thus turn predicates on the set $X$ of states into
predicates on the set $TX$ of structured collections of successors. A
basic example is the predicate lifting for the usual diamond modality
$\Diamond$, given by
$\Sem{\Diamond}_X(A)=\{B\in\Pow X\mid B\cap A\neq\emptyset
\rbrace$. We will see more examples in Sections~\ref{sec:presburger}
and~\ref{sec:prob}. For purposes of the generic technical development,
\emph{we fix the data $\Lambda$,~$T$, and~$\Sem{\hearts}$ throughout,
  and by abuse of notation sometimes refer to them jointly as (the
  logic)~$\Lambda$.}

Satisfaction $x\models_C\phi$ (or just $x\models\phi$ when~$C$ is
clear from the context) of formulae $\phi\in\FLang(\Lambda)$ in
states~$x$ of a coalgebra $C=(X,\gamma)$ is defined inductively by
\begin{align*}
x& \not\models_C\bot \\
x& \models_C\phi\land\psi && \hspace{-4em}\text{iff}\quad x\models_C\phi\text{ and }x\models_C\psi\\
x& \models_C\neg\phi &&\hspace{-4em} \text{iff}\quad x\not\models_C\phi\\
x&\models_C \hearts(\phi_1,\dots,\phi_n)&&\hspace{-4em} 
  \text{iff}
  \quad
  \gamma(x)\models\hearts(\Sem{\phi_1}_C,\dots,\Sem{\phi_n}_C)
\end{align*}
where we write $\Sem{\phi}_C=\{x\in X\mid x\models_C\phi\}$ (and use
notation as per~\eqref{eq:lifting-notation}). Continuing the above
example, the predicate lifting $\Sem{\Diamond}$ thus induces exactly
the usual semantics of $\Diamond$: Given a $\Pow$-coalgebra, i.e.\
Kripke frame, $(X,\gamma\colon X\to\Pow X)$, we have
$x\models_C\Diamond\phi$ iff the set $\gamma(x)$ of successors of $x$
intersects with~$\Sem{\phi}_C$, i.e.\ iff~$x$ has a successor that
satisfies~$\phi$.

We will be interested in satisfiability under global assumptions, or,
in description logic terminology, reasoning with general
TBoxes~\cite{BaaderEA03}, that is, under background axioms that are
required to hold in every state of a model:
\begin{definition}[Global assumptions]
  Given a formula~$\psi$, the \emph{global assumption}, a coalgebra
  $C=(X,\gamma)$ is a \emph{$\psi$-model} if $\Sem{\psi}_C=X$; and a
  formula $\phi$ is \emph{$\psi$-satisfiable} if there exists a
  $\psi$-model $C$ such that $\Sem{\phi}_C\neq\emptyset$. The
  \emph{satisfiability problem under global assumptions} is to decide,
  given~$\psi$ and~$\phi$, whether~$\phi$ is $\psi$-satisfiable. We
  extend these notions to sets~$\Gamma$ of formulae: We write
  $x\models_C\Gamma$ if $x\models\phi$ for all $\phi\in\Gamma$, and we
  say that $\Gamma$ is \emph{$\psi$-satisfiable} if there exists a
  state~$x$ in a $\psi$-model~$C$ such that $x\models_C\Gamma$. For
  distinction, we will occasionally refer to satisfiability in the
  absence of global assumptions, i.e. $\top$-satisfiability, as
  \emph{plain satisfiability}.
\end{definition}
\noindent
\begin{rem}
  While the typical complexity of plain satisfiability is \PSpace,
  that of satisfiability under global assumptions is \ExpTime. In
  particular, this holds for the basic example of relational modal
  logic~\cite{Ladner77,FischerLadner79}.

  As indicated above, global assumptions are referred to as \emph{TBox
    axioms} in description logic parlance, in honour of the fact that
  they capture what is, in that context, called \emph{terminological
    knowledge}: They record facts that hold about the world at large,
  such as `every car has a motor' (formalized, e.g., in relational
  modal logic as $\psi:=(\mathsf{Car}\to\Diamond\,\mathsf{Motor})$ if
  the relation that underlies~$\Diamond$ is understood as
  parthood). Contrastingly, a formula~$\phi$ is satisfiable under the
  global assumption~$\psi$ as soon as $\phi$ holds in \emph{some}
  state of some $\psi$-model, so~$\phi$ is thought of as describing
  some states (\emph{individuals} in description logic terminology)
  but not as being universally true. Correspondingly, the reasoning
  task of checking satisfiability (under global assumptions) is called
  \emph{concept satisfiability (under general TBoxes)} in description
  logic. For instance, the atomic proposition (`concept')
  $\mathsf{Car}$ is $\psi$-satisfiable in the above example, but not
  of course necessarily true in every state of a $\psi$-model.

  \emph{Global consequence}, i.e.~entailment between global
  assumptions, reduces to satisfiability under global assumptions: We
  say that a formula~$\phi$ is a \emph{global consequence} of a
  formula~$\psi$ if every $\psi$-model is also a $\phi$-model. Then
  $\phi$ is a global consequence of~$\psi$ iff $\neg\phi$ is not
  $\psi$-satisfiable. For instance, in relational modal logic,
  $\Box\psi$ is always a global consequence of~$\psi$,
  i.e.~$\neg\Box\psi$ is not $\psi$-satisfiable; this fact corresponds
  to the well-known necessitation rule of relational modal
  logic~\cite{BlackburnEA01}.
\end{rem}

\begin{rem}\label{rem:univ-mod}
  As indicated in the introduction, for purposes of the complexity
  analysis, global assumptions are equivalent to the universal
  modality. We make this claim more precise as follows. We define
  \emph{coalgebraic modal logic with the universal modality} by
  extending the grammar~\eqref{eq:grammar} with an additional
  alternative
  \begin{equation*}
    \dots \mid \univbox \phi,
  \end{equation*}
  and the semantics with the clause
  \begin{equation*}
    x\models_C \univbox\phi\quad\text{iff}\quad y\models_C\phi\text{ for all $y\in X$}
  \end{equation*}
  for a coalgebra $C=(X,\gamma)$. In this logic, we restrict attention
  to plain satisfiability checking, asking whether, for a given
  formula~$\phi$, there exists a state~$x$ in a coalgebra~$C$ such
  that $x\models_C\phi$. Then satisfiability under global assumptions
  clearly reduces in logarithmic space to plain satisfiability in
  coalgebraic modal logic with the universal modality -- a
  formula~$\phi$ is satisfiable under the global assumption~$\psi$ iff
  $\phi\land\univbox\psi$ is satisfiable.

  Conversely, satisfiability of a formula~$\phi$ in coalgebraic modal
  logic with the universal modality is reducible in nondeterministic
  polynomial time to satisfiability under global assumptions in
  coalgebraic modal logic, as follows. Call a subformula of~$\phi$ a
  \emph{$\univbox$-subformula} if it is of the shape $\univbox\psi$,
  and let $\univbox\psi_1,\dots,\univbox\psi_n$ be the
  $\univbox$-subformulae of~$\phi$. Given a subset
  $U\subseteq\{1,\dots,n\}$ and a subformula~$\chi$ of~$\phi$, denote
  by $\chi[U]$ the $\univbox$-free formula obtained from~$\chi$ by
  replacing every $\univbox$-subformula $\univbox\psi_k$ that is not
  in scope of a further~$\univbox$-operator by~$\top$ if $k\in U$, and
  by $\bot$ otherwise.  We claim that
  \begin{quote}
    ($*$) $\phi$ is satisfiable (in coalgebraic modal logic with the
    universal modality) iff there is~$U\subseteq\{1,\dots,n\}$ such
    that~$\phi[U]$, as well as each formula~$\neg\psi_k[U]$ for
    $k\in\{1,\dots,n\}\setminus U$, are (separately) satisfiable under
    the global assumption~$\psi_U$ given by
    $\psi_U=\Land_{k\in U}\psi_k[U]$.
  \end{quote}
  Using~($*$), we can clearly reduce satisfiability in coalgebraic
  modal logic with the universal modality to satisfiability under
  global assumptions in coalgebraic modal logic as claimed by just
  guessing~$U$. It remains to prove~($*$). For the `only if'
  direction, suppose that $x\models_C\phi$ for some state~$x$ in a
  $T$-coalgebra~$C=(X,\gamma)$. Put
  $U=\{k\mid x\models_C\univbox\psi_k\}$. It is readily checked that,
  in the above notation,~$C$ is a $\psi_U$-model, $x\models_C\phi[U]$,
  and for each $k\in\{1,\dots,n\}\setminus U$, $\neg\psi_k[U]$ is
  satisfied in some state of~$C$. For the converse implication, let
  $U\subseteq\{1,\dots,n\}$, let $C$ and $C_k$, for
  $k\in\{1,\dots,n\}\setminus U$, be $\psi_U$-models, let
  $x\models_C\phi[U]$, and let $x_k\models_{C_k}\neg\psi_k[U]$ for
  $k\in\{1,\dots,n\}\setminus U$. Let $D$ be the disjoint union of~$C$
  and the~$C_k$; it is straightforward to check that $x\models_D\phi.$

  It follows that from the exponential-time upper bound for
  satisfiability checking under global assumptions proved in
  Section~\ref{sec:type-elim}, we obtain an exponential-time upper
  bound for satisfiability checking in coalgebraic modal logic with
  the universal modality. On the other hand, the non-deterministic
  reduction described above of course does not allow for inheriting
  practical reasoning algorithms. The design of tableau-based
  algorithms in presence of the universal modality is faced with the
  challenge that instances of $\univbox$ uncovered deep in the formula
  by the rule-based decomposition will subsequently influence the
  entire tableau built so far. Our global caching algorithm
  (Section~\ref{sec:caching}) is meant for reasoning under global
  assumptions; we leave the design of a practical generic reasoning
  algorithm for coalgebraic modal logic with the universal modality to
  future work.
\end{rem}
\noindent Generic algorithms in coalgebraic logic frequently rely on
complete rule sets for the given modal
operators~\cite{SchroderPattinson09a} (an overview of the relevant
concepts is given in Remark~\ref{rem:rules}); in particular, such a
rule set is assumed by our previous algorithm for satisfiability
checking under global assumptions in coalgebraic hybrid
logic~\cite{SchroderEA09}. In the present paper, our interest is in
cases for which suitable rule sets are not (currently) available. We
proceed to present our leading examples of this kind, Presburger modal
logic and a probabilistic modal logic with polynomial
inequalities. For the sake of readability, we focus on the case with a
single (weighted) transition relation, and omit propositional
atoms. Both propositional atoms and indexed transition relations are
easily added, e.g.\ using compositionality results in coalgebraic
logic~\cite{SchroderPattinson11MSCS}, and in fact we use them freely
in the examples; more details on this point will be provided in
Remark~\ref{rem:atoms}.

\subsection{Presburger Modal Logic}
\label{sec:presburger}

Presburger modal logic~\cite{DemriLugiez10} admits statements in
Presburger arithmetic over numbers $\sharp\phi$ of successors
satisfying a formula $\phi$. Throughout, we let $\Rels$ denote the set
$\{<,>,=\}\cup\{\eqmod{k}\mid k\in\Nat\}$ of \emph{arithmetic
  relations}, with $\eqmod{k}$ read as congruence modulo
$k$. Syntactically, Presburger modal logic is then defined in our
syntactic framework by taking the modal similarity type
\begin{equation*}
  \Lambda =\{L_{u_1,\dots,u_n;\sim \inhom}\mid
  {\sim}\in\Rels, n\in\Nat, u_1,\dots,u_n,\inhom\in\ZZ\}
\end{equation*}
where $L_{u_1,\dots,u_n;\sim \inhom}$ has arity~$n$. The application
of a modal operator $L_{u_1,\dots,u_n;\sim \inhom}$ to argument
formulae $\phi_1,\dots,\phi_n$ is written
\begin{equation*}\textstyle
  \textstyle \sum_{i=1}^nu_i\cdot\sharp \phi_i\sim \inhom.
\end{equation*}
We refer to these modalities as \emph{Presburger constraints}. Weak
inequalities can be coded as strict ones, replacing, e.g., $\ge k$
with $>k-1$. The numbers $u_i$ and $\inhom$, as well as the modulus
$k$ in $\eqmod{k}$, are referred to as the \emph{coefficients} of a
Presburger constraint. We also apply this terminology (Presburger
constraint, coefficient) to constraints of the form
$\sum_{i=1}^nu_i\cdot x_i\sim \inhom$ in general, interpreted over the
non-negative integers.

The semantics of Presburger modal logic was originally defined over
standard Kripke frames; in order to make sense of sums with arbitrary
(possibly negative) integer coefficients, one needs to
restrict to finitely branching frames. We consider an alternative
semantics in terms of \emph{multigraphs}, which have some key
technical advantages~\cite{DAgostinoVisser02}. Informally, a
multigraph is like a Kripke frame but with every transition edge
annotated with a non-negative-integer-valued multiplicity; ordinary
finitely branching Kripke frames can be viewed as multigraphs by just
taking edges to be transitions with multiplicity~$1$. Formally, a
multigraph can be seen as a coalgebra for the \emph{finite multiset
  functor}~$\Bag$: For a set $X$, $\Bag X$ consists of the
\emph{finite multisets over $X$}, which are maps $\mu\colon X\to\Nat$
with finite support, i.e.\ $\mu(x)>0$ for only finitely many $x$.  We
view $\mu$ as an $\Nat$-valued measure, and write
$\mu(Y)=\sum_{x\in Y}\mu(x)$ for $Y\subseteq X$.  Then, $\Bag f$, for
maps $f\colon X\to Y$, acts as image measure formation in the same way
as the distribution functor~$\Dist$ described above, i.e.\
$(\Bag f)(\mu)(B)=\mu(f^{-1}[B])$ for $\mu\in\Bag X$ and
$B\subseteq Y$. A coalgebra $\gamma\colon X\to\Bag X$ assigns to each
state $x$ a multiset $\gamma(x)$ of successor states, i.e.\ each
successor state is assigned a transition multiplicity.

The semantics of the modal operators is then given by the predicate liftings
\begin{equation*}\textstyle
  \Sem{L_{u_1,\dots,u_n;\sim \inhom}}_X(A_1,\dots,A_n) =
  \{\mu\in\Bag X \mid \sum_{i=1}^n u_i \cdot \mu(A_i) \sim
    \inhom\},
\end{equation*}
that is, a state $x$ in a $\Bag$-coalgebra $C=(X,\gamma)$ satisfies
$\sum_{i=1}^n u_i \cdot \sharp\phi_i\sim \inhom$ iff
$\sum_{i=1}^n u_i \cdot \gamma(x)(\Sem{\phi_i}_C)\sim
\inhom$.

\begin{rem}
  \emph{Graded modal logic}~\cite{Fine72} is interpreted over the same
  systems (originally Kripke frames, equivalently multigraphs) as
  Presburger modal logic. It combines a Boolean propositional base
  with modalities $\gldiamond{k}$ `in more than~$k$ successors'; these
  have made their way into modern expressive description logics in the
  shape of \emph{qualified number restrictions}~\cite{BaaderEA03}. The
  multigraph semantics of graded modal logic is captured
  coalgebraically by assigning to $\gldiamond{k}$ the predicate
  lifting for~$\Bag$ given by
  $\Sem{\gldiamond{k}}_X(A)=\{\mu\in\Bag(X)\mid\mu(A)>k\}$.
  Presburger modal logic subsumes graded modal logic, via a
  translation~$t$ of graded modal logic into Presburger modal logic
  that is defined by commutation with all Boolean connectives and
  $t(\gldiamond{k}\phi)=(\sharp(t(\phi))>k)$.
\end{rem}
\noindent We note that satisfiability is the same over Kripke frames
and over multigraphs:

\begin{lemma}~\cite[Remark~6]{Schroder07} \cite[Lemma 2.4]{SchroderVenema18} \label{lem:multi-vs-kripke}
  A formula $\phi$ is $\psi$-satisfiable over multigraphs iff $\phi$
  is $\psi$-satisfiable over Kripke frames.
\end{lemma}
\noindent (The proof of the non-trivial direction is by making copies
of states to accommodate multiplicities.)
\begin{rem}
  From the point of view of the present work, the technical reason to
  work with multigraphs rather than Kripke frames in the semantics of
  Presburger modal logic is that the key naturality
  condition~\eqref{eq:naturality} fails over Kripke semantics, i.e.\
  for the powerset functor. Beyond the mere fact that for this reason,
  our methods do not apply to the Kripke semantics of Presburger or
  graded modal logic, we note that indeed key results of coalgebraic
  modal logic fail to hold for this semantics. For instance, we shall
  prove later (Lemma~\ref{lem:one-step-models}) that coalgebraic modal
  logic has the exponential model property, i.e.\ every satisfiable
  formula~$\phi$ has a model with at most exponentially many states in
  the number of subformulae of~$\phi$. Over Kripke semantics, this
  clearly fails already for simple formulae such as~$\sharp(\top)>k$.
\end{rem}
\begin{rem}\label{rem:atoms}
  As indicated above, the overall setup generalizes effortlessly to
  allow for both propositional atoms and multiple (weighted)
  transition relations: Let~$\At$ be a set of \emph{propositional
    atoms} and~$\Roles$ a set of \emph{relation names} (\emph{atomic
    concepts} and \emph{roles}, respectively, in description logic
  terminology). We then take the modal operators to be the
  propositional atoms and all operators
  \begin{equation*}\textstyle
    L_{u_1^{r_1},\dots,u_n^{r_n};\sim \inhom}
  \end{equation*}
  where ${\sim}\in\Rels$, $n\in\Nat$, $u_1,\dots,u_n,\inhom\in\ZZ$,
  and $r_1,\dots,r_n \in \Roles$. The arity of
  $L_{u_1^{r_1},\dots,u_n^{r_n};\sim \inhom}$ is~$n$, and the
  application of $L_{u_1^{r_1},\dots,u_n^{r_n};\sim \inhom}$ to
  argument formulae $\phi_1,\dots,\phi_n $ is written
  \begin{equation*}
    \textstyle
    \sum_{i=1}^nu_i\cdot\sharp_{r_i} \phi_i\sim \inhom
  \end{equation*}
  where $\sharp_r(\cdot)$ is meant to represent the number of
  successors along the (weighted) transition relation $r$. The logic
  is then interpreted over structures that assign to each state~$x$ a
  subset of~$\At$ (of propositional atoms that hold at~$x$) and
  $\Roles$-many multisets of successors. Such structures as coalgebras
  for the functor that maps a set~$X$ to
  $\Pow \At\times \Bag X^\Roles$; the associated predicate liftings
  are given by
  \begin{align*}
    \Sem{L_{u_1^{r_1},\dots,u_n^{r_n};\sim \inhom}}_X(A_1,\dots,A_n) & =
                                                                  \{(U,f)\in\Pow\At\times (\Bag X)^\Roles \mid \textstyle\sum_{i=1}^n u_i \cdot f(r_i)(A_i) \sim
                                                                  \inhom\}\\
    \Sem{p}_X & = \{(U,f)\in\Pow\At\times (\Bag X)^\Roles \mid p\in U\}.
  \end{align*}
  The effect of these extensions on the technical development does not
  go beyond heavier notation, so as announced above we restrict the
  exposition to only a single transition relation and no propositional
  atoms, for readability.
\end{rem}
\begin{rem}\label{rem:graded-rules}
  Two of us (Kupke and Pattinson) have exhibited modal sequent rules
  for various modal logics of linear inequalities, both over the
  non-negative reals (e.g.\ probabilistic and stochastic logics) and
  over the non-negative integers~\cite{Kupke:2010:MLL}. One of these
  logics can be seen as the fragment of Presburger modal logic
  obtained by removing modular congruence~$\equiv_k$. Soundness and
  completeness of the rules for this logic would imply our upper
  complexity bounds by instantiating our own previous generic results
  in coalgebraic logic~\cite{SchroderEA09}, which rely on precisely
  such rules. However, while the rules given for logics with
  real-valued multiplicities appear to be sound and complete as
  claimed, the rule system given for the integer-valued case is sound
  but clearly not complete, as indicated already in
  Section~\ref{sec:intro}. For instance, the formula
  $\phi:=(2\sharp\top<1\lor 2\sharp\top>1)$ is valid for integer
  multiplicities ($\phi$ says that the integer total weight of all
  successors of a state cannot be $1/2$) but not provable in the given
  rule system. The latter fact is most easily seen by comparing the
  rule for integer multiplicities \cite[Section~4]{Kupke:2010:MLL}
  with the rule given for the case of real-valued multiplicities
  \cite[Section~5]{Kupke:2010:MLL}: The rule instances applying
  to~$\phi$ are the same in both cases, and as the rules are easily
  seen to be sound in the real-valued case,~$\phi$ is not provable (as
  it fails to be valid in the real-valued case). There does not seem
  to be an easy fix for this, so for the time being there is no known
  sound and complete set of modal sequent rules (equivalently, modal
  tableau rules) for Presburger modal logic.
\end{rem}

\smallskip\noindent\textbf{Expressiveness and Examples.}  As mentioned
above, Presburger modal logic subsumes graded modal
logic~\cite{Fine72}.  Moreover, Presburger modal logic subsumes
majority logic~\cite{PacuitSalame04} (more precisely, the version of
majority logic interpreted over finitely branching systems): The
\emph{weak majority} formula $W\phi$ (`at least half the successors
satisfy $\phi$') is expressed in Presburger modal logic as
$\sharp(\phi)-\sharp(\neg\phi)\ge 0$. Using propositional atoms,
incorporated in the way discussed above, we express the examples given
in the introduction (`the majority of university students are female',
`dance classes have even numbers of participants') by the formulae
\begin{gather*}
  \ms{University} \to (\sharp_\ms{hasStudent}\ms{Female}
  - \sharp_\ms{hasStudent}\ms{Male}>0)\\
  \ms{DanceCourse} \to (\sharp_\ms{hasParticipant}\ms{\top}
  \equiv_2 0)
\end{gather*}
where indices informally indicate the understanding of the successor
relation.  In the extension with multiple successor relations
(Remark~\ref{rem:atoms}), one may also impose inequalities between
numbers of successors under different roles as in the introduction,
e.g.~in the formula
\begin{equation*}
  \ms{Workaholic} \to (\rcount{\ms{hasColleague}}{\top}-  \rcount{\ms{hasFriend}}{\top} >0)
\end{equation*}
(`workaholics have more colleagues than friends'). As an example
involving non-unit coefficients, a chamber of parliament in which a
motion requiring a 2/3 majority has sufficient support is described by
the formula
\begin{equation*}
  \sharp_{\ms{hasMember}}(\ms{SupportsMotion})- 
    2\sharp_{\ms{hasMember}}(\neg\ms{SupportsMotion})\ge 0.
\end{equation*}

\subsection{Probabilistic Modal Logic with Polynomial Inequalities}
\label{sec:prob}

Probabilistic logics of various forms have been studied in different
contexts such as reactive systems~\cite{LarsenSkou91} and uncertain
knowledge~\cite{HeifetzMongin01,FaginHalpern94}. A typical feature of
such logics is that they talk about probabilities $w(\phi)$ of
formulae $\phi$ holding for the successors of a state; the concrete
syntax then variously includes only inequalities of the form
$w(\phi)\sim p$ for ${\sim}\in\{>,\ge,=,<,\le\}$ and
$p\in\mathbb{Q}\cap[0,1]$~\cite{LarsenSkou91,HeifetzMongin01}, linear
inequalities over terms $w(\phi)$~\cite{FaginHalpern94}, or polynomial
inequalities, with the latter so far treated only in either purely
propositional settings~\cite{FaginHalpernMegiddo90} or in
many-dimensional logics such as the probabilistic description logic
Prob-$\ALC$~\cite{GutierrezBasultoEA17}, which use a single global
distribution over worlds. An important use of polynomial inequalities
over probabilities is to express independence
constraints~\cite{GutierrezBasultoEA17}. For instance, two properties $\phi$
and $\psi$ (of successors) are independent if
$w(\phi\land\psi)=w(\phi)w(\psi)$, and we can express that the
probability that the first of two independently sampled successors
satisfies~$\phi$ and the second satisfies~$\psi$ is at least~$p$ by a
formula such as $w(\phi)w(\psi)\ge p$; the latter is similar to the
\emph{independent product} of real-valued probabilistic modal
logic~\cite{Mio11}.

We thus define the following \emph{probabilistic modal logic with
  polynomial inequalities}: The system type is given by a variant of
the distribution functor~$\Dist$ as described above, viz.\ the
\emph{subdistribution functor}~$\SDist$, in which we require for
$\mu\in\SDist X$ that the measure of the whole set~$X$ satisfies
$\mu(X)\le 1$ rather than $\mu(X)=1$. Then $\SDist$-coalgebras
$\gamma:X\to\SDist X$ are like Markov chains (where $\gamma(x)$ is
interpreted as a distribution over possible future evolutions of the
system), or (single-agent) type spaces in the sense of epistemic
logic~\cite{HeifetzMongin01} (where $\gamma(x)$ is interpreted as the
subjective probabilities assigned by the agent to possible alternative
worlds in world $x$), with the difference that each state~$x$ has a
probability $1-\gamma(x)(X)$ of being deadlocked. We use the modal
similarity type
\begin{equation*}
  \Lambda=\{L_p\mid p\in\mathbb{Q}[X_1,\dots,X_n], n\ge 0\};
\end{equation*}
for $p\in\mathbb{Q}[X_1,\dots,X_n]$, the modality~$L_p$ has
arity~$n$. We denote the application of~$L_p$ to formulae
$\phi_1,\dots,\phi_n$ by substituting each variable $X_i$ in $p$ with
$w(\phi_i)$ and postulating the result to be non-negative, i.e.~as
\begin{equation*}
  p(w(\phi_1),\dots,w(\phi_n))\ge 0.
\end{equation*}
For instance, $L_{X_1-X_2X_3}(\phi\land\psi,\phi,\psi)$ is written
more readably as $w(\phi\land\psi)-w(\phi)w(\psi)\ge 0$, and thus
expresses one half of the above-mentioned independence constraint (the
other half, of course, being $w(\phi)w(\psi)-w(\phi\land\psi)\ge 0$)
We correspondingly interpret $L_p$ by the predicate lifting
\begin{equation*}
  \Sem{L_p}_X(A_1,\dots,A_n)=\{\mu\in\SDist X\mid p(\mu(A_1),\dots,\mu(A_n))\ge 0\}.
\end{equation*}
We will use Presburger modal logic and probabilistic modal logic as
running examples in the sequel.

\begin{rem}
  The use of~$\SDist$ in place of~$\Dist$ serves only to avoid
  triviality of the logic in the absence of propositional atoms: Since
  $|\Dist(1)|=1$ for any singleton set~$1$, all states in
  $\Dist$-coalgebras (i.e.\ Markov chains) are bisimilar, and thus
  satisfy the same formulae of any coalgebraic modal logic on
  $\Dist$-coalgebras~\cite{Pattinson04,Schroder08}, so any formula in
  such a logic is either valid or unsatisfiable. This phenomenon
  disappears as soon as we add propositional atoms as per
  Remark~\ref{rem:atoms}. All our results otherwise apply to~$\Dist$
  in the same way as to~$\SDist$.
\end{rem}

\section{One-Step Satisfiability}\label{sec:oss}
The key ingredient of our algorithmic approach is to deal with modal
operators (i.e., in our running examples, arithmetic statements about
numbers or weights of successors) level by level; the core concepts of
the arising notion of one-step satisfiability checking go back to work
on plain satisfiability in coalgebraic
logics~\cite{Schroder07,SchroderPattinson08d,MyersEA09}. From now on,
we restrict the technical treatment to unary modal operators to avoid
cumbersome notation, although our central examples all do have modal
operators with higher arities; a fully general treatment requires no
more than additional indexing.
Considering only one level of modal operators and abstracting from
their arguments amounts to working in a \emph{one-step logic}, whose
syntax and semantics are defined as follows (subsequent to fixing some
notation).
\begin{definition}[Notation for propositional variables and
  propositional logic]\label{def:prop}
  We fix a countably infinite set $\PV$ of \emph{(propositional)
    variables}. We denote the set of Boolean formulae (presented in
  terms of~$\bot$,~$\land$, and~$\neg$) over a set~$V\subseteq\PV$ of
  propositional variables by $\Prop(V)$; that is, formulae
  $\eta,\rho\in\Prop(V)$ are defined by the grammar
  \begin{equation*}
    \eta,\rho::=\bot\mid\neg\eta\mid\eta\land\rho\mid a \qquad (a\in V).
  \end{equation*}
  We write~$2$ for the set $\{\bot,\top\}$ of truth values, and then
  have a standard notion of satisfaction of propositional formulae
  over~$V$ by valuations $\kappa\colon V\to 2$. As usual, a
  \emph{literal} over $V$ is a propositional variable $a\in V$ or a
  negated variable $\neg a$ for $a\in V$, often written~$\epsilon a$
  with $\epsilon\in\{-1,1\}$ as per the previous convention
  (Section~\ref{sec:colog}), and a \emph{conjunctive clause} over~$V$
  is a finite conjunction $\epsilon_1a_1\land\dots\land\epsilon_na_n$
  of literals over~$V$, represented as a finite set of literals. We
  write $\Phi\PLentails\eta$ to indicate that a set
  $\Phi\subseteq\Prop(V)$ \emph{propositionally entails}
  $\eta\in\Prop(V)$, meaning that there exist
  $\rho_1,\dots,\rho_n\in\Phi$ such that
  $\rho_1\land\dots\land\rho_n\to\eta$ is a propositional
  tautology. For $\{\rho\}\PLentails\eta$, we briefly write
  $\rho\PLentails\eta$.

  By a \emph{substitution}, we will mean a map~$\sigma$ from (some
  subset of)~$\PV$ into another set~$Z$, typically a set of formulae
  of some kind. In case $Z=\PV$, we will also refer to~$\sigma$ as a
  \emph{renaming}. We write application of a substitution~$\sigma$ to
  formulae~$\phi$ containing propositional variables (either
  propositional formulae or formulae of the one-step logic as
  introduced in the next definition) in postfix notation $\phi\sigma$
  as usual (i.e.~$\phi\sigma$ is obtained from~$\phi$ by replacing all
  occurrences of propositional variables~$a$ in~$\phi$ with
  $\sigma(a)$). We extend the propositional entailment relation to
  formulae beyond~$\Prop(\PV)$ by substitution, i.e.~for a
  formula~$\psi$ and a set~$\Phi$ of formulae (in the one-step logic
  or in coalgebraic modal logic), we write $\Phi\PLentails\psi$ if
  $\Phi,\psi$ can be written in the form $\Phi=\Phi'\sigma$,
  $\psi=\psi'\sigma$ for a substitution~$\sigma$ and
  $\Phi'\subseteq\Prop(\PV)$, $\psi'\in\Prop(\PV)$ such that
  $\Phi'\PLentails\psi'$ in the sense defined above (that is, if there
  are $\phi_1,\dots,\phi_n\in\Phi$ such that
  $\phi_1\land\dots\land\phi_n\to\psi$ is a substitution instance of a
  propositional tautology).
\end{definition}
\noindent The \textbf{syntax} of the one-step logic is given in the following
terms:
\begin{definition}[One-step pairs]\label{def:one-step}
  Given a set~$V\subseteq\PV$ of propositional variables, we denote by
  \begin{equation*}
    \Lambda(V)=\{\hearts a\mid\hearts\in\Lambda,a\in V\}
  \end{equation*}
  the set of \emph{modal atoms over~$V$}.  A \emph{modal literal
    over~$V$} is a modal atom over~$V$ or a negation thereof, i.e.~has
  the form either $\hearts a$ or $\neg\hearts a$ for
  $\hearts\in\Lambda$, $a\in V$. A \emph{modal conjunctive
    clause}~$\phi$ is a finite conjunction
  $\epsilon_1\hearts_1 a_1\land\dots\land\epsilon_n\hearts_na_n$ of
  modal literals over~$V$, represented as a finite set of modal
  literals. We write $\Var(\phi)=\{a_1,\dots,a_n\}$ for the set of
  variables occurring in~$\phi$. We say that~$\phi$ is \emph{clean}
  if~$\phi$ mentions each variable in~$V$ at most once. A
  \emph{one-step pair} $(\phi,\eta)$ \emph{over $V$} consists of
  \begin{myitemize}
  \item a clean modal conjunctive clause~$\phi$ over~$V$ and
  \item a Boolean formula $\eta\in\Prop(\Var(\phi))$.
  \end{myitemize}
  We measure the size~$|\phi|$ of a modal conjunctive clause~$\phi$ by
  counting~$1$ for each variable and each propositional operator, and
  for each modality the size of its encoding (in the same way as in
  the definition of the size of modal formulae in
  Section~\ref{sec:colog}). The propositional component~$\eta$ is
  assumed to be given as a DNF consisting of conjunctive clauses each
  mentioning every variable occurring in~$\phi$ (such conjunctive
  clauses are effectively truth valuations for the variables
  in~$\phi$), and the size~$|\eta|$ of~$\eta$ is the size of this DNF.
\end{definition}
\noindent In a one-step pair $(\phi,\eta)$, the modal component~$\phi$
effectively specifies what happens one transition step ahead from the
(implicit) current state; as indicated above, in the actual
satisfiability checking algorithm,~$\phi$ will arise by peeling off
the top layer of modalities of a given modal formula, with the
propositional variables in~$V$ abstracting the argument formulae of
the modalities. The propositional component~$\eta$ then records the
propositional dependencies among the argument formulae. Formally, the
\textbf{semantics} of the one-step logic is given as follows:
\begin{definition}[One-step models, one-step satisfiability] A
  \emph{one-step model} $M=(X,\tau, t)$ over $V$ consists of
  \begin{myitemize}
  \item a set $X$ together with a $\Pow X$-valuation $\tau\colon V\to\Pow X$; and
  \item an element $t\in TX$ (thought of as the structured collection
    of successors of an anonymous state).
  \end{myitemize}
  For $\eta\in\Prop(V)$, we write $\tau(\eta)$ for the interpretation
  of $\eta$ in the Boolean algebra $\Pow X$ under the valuation
  $\tau$; explicitly, $\tau(\bot)=\emptyset$,
  $\tau(\neg\eta)=X\setminus\tau(\eta)$, and
  $\tau(\eta\land\rho)=\tau(\eta)\cap\tau(\rho)$.  For a modal atom
  $\hearts a\in\Lambda(V)$, we put
  \begin{equation*}
    \tau(\hearts a)=\Sem{\hearts}_X(\tau(a))\subseteq TX.
  \end{equation*}
  We extend this assignment to modal atoms and modal conjunctive
  clauses using the Boolean algebra structure of $\Pow(TX)$; explicitly,
  \begin{align*}
    \tau(\neg\hearts a)
    &=TX\setminus\tau(\hearts a)\\
    \tau(\epsilon_1\hearts_1a_1\land\dots\land\epsilon_n\hearts_na_n)
    &=\tau(\epsilon_1\hearts_1a_1)\cap\dots\cap\tau(\epsilon_n\hearts_na_n).
  \end{align*}
  We say that the one-step model $M=(X,\tau,t)$ \emph{satisfies} the
  one step pair $(\phi,\eta)$, and write $M\models(\phi,\eta)$, if
  \begin{equation*}
    \tau(\eta)=X\qquad\text{and}\qquad t\in\tau(\phi).
  \end{equation*}
  (That is,~$\eta$ is a global propositional constraint on the values
  of~$\tau$ while~$\phi$ specifies a property of the collection~$t$ of
  successors.) Then, $(\phi,\eta)$ is \emph{(one-step) satisfiable} if
  there exists a one-step model~$M$ such that
  $M\models(\phi,\eta)$. The \emph{lax one-step satisfiability
    problem} (of~$\Lambda$) is to decide whether a given one-step pair
  $(\phi,\eta)$ is one-step satisfiable; the size of the input is
  measured as $|\phi|+|\eta|$ with~$|\phi|$ and~$|\eta|$ defined as
  above. The \emph{strict one-step satisfiability problem}
  (of~$\Lambda$) is the same problem but with the input size defined
  to be just~$|\phi|$. For purposes of space complexity, we thus
  assume in the strict one-step satisfiability problem that~$\eta$ is
  stored on an input tape that does not count towards space
  consumption. It will be technically convenient to assume moreover
  that in the strict one-step satisfiability problem,~$\eta$ is given
  as a bit vector indicating which conjunctive clauses (mentioning
  every variable occurring in $\phi$, in some fixed order) are
  contained in the DNF~$\eta$; contrastingly, we assume that in the
  lax one-step satisfiability problem,~$\eta$ is given as a list of
  conjunctive clauses as indicated in Definition~\ref{def:one-step}
  (hence need not have exponential size in all cases). For time
  complexity, we assume that the input tape is random access (i.e.\
  accessed via a dedicated address tape, in the model of random access
  Turing machines~\cite{FischerRosenberg68}; this is necessary to
  enable subexponential time bounds for the strict one-step
  satisfiability problem since otherwise it takes exponential time
  just to move the head to the last bits of the input). We say
  that~$\Lambda$ has the \emph{(weak) one-step small model property}
  if there is a polynomial~$p$ such that every one-step satisfiable
  $(\phi,\eta)$ has a one-step model $(X,\tau,t)$ with
  $|X|\le p(|\Var(\phi)|)$ (respectively $|X|\le p(|\phi|)$). (Note
  that no bound is assumed on the representation of $t$.)
\end{definition}
\noindent
As indicated above, the intuition behind these definitions is that the
propositional variables in~$V$ are placeholders for argument formulae
of modalities; their valuation $\tau$ in a one-step model $(X,\tau,t)$
over~$V$ represents the extensions of these argument formulae in a
model; and the second component~$\eta$ of a one-step pair
$(\phi,\eta)$ captures the Boolean constraints on the argument
formulae that are globally satisfied in a given model. The component
$t\in TX$ of $(X,\tau,t)$ represents the structured collection of
successors of an implicit current state, so the modal component~$\phi$
of the one-step pair is evaluated on~$t$. 
We will later construct
full models of modal formulae using one-step models according to this
intuition. One may think of a one-step model $(X, \tau, \mu)$ of a
one-step pair $(\phi, \eta)$ as a counterexample to soundness of
$\eta/\neg\phi$ as a proof rule:~$\phi$ is satisfiable despite $\eta$
being globally valid in the model.
\begin{example}\label{expl:oss}
  \begin{enumerate}[wide]
  \item In the basic example of relational modal logic
  ($\Lambda=\{\Diamond\}$, $T=\Pow$, see Section~\ref{sec:colog}),
  consider the one-step pair
  $(\phi,\eta):=(\neg\Diamond a\land\neg\Diamond b\land \Diamond c,c\to
  a\lor b)$. The propositional component~$\eta$ is represented as a
  DNF $\eta=(c\land a\land b)\lor(\neg c\land\neg a\land\neg b)\lor\dots$.
  A one-step model $(X,\tau,t)$ of $(\phi,\eta)$ (where $t\in\Pow(X)$)
  would need to satisfy $\tau(c)\subseteq\tau(a)\cup\tau(b)$ to ensure
  $\tau(\eta)=X$, as well as $t\cap\tau(c)\neq\emptyset$,
  $t\cap \tau(a)=\emptyset$, and $t\cap \tau(b)=\emptyset$ to ensure
  $t\in\tau(\phi)$. As this is clearly impossible, $(\phi,\eta)$ is
  unsatisfiable. In fact, it is easy to see that the strict one-step
  satisfiability problem of relational modal logic in this sense is in
  \NP: To check whether a one-step pair $(\psi,\chi)$ is satisfiable,
  guess a conjunctive clause~$\rho$ in~$\chi$ for each positive modal
  literal $\Diamond a$ in~$\phi$, and check that $\rho$ contains on
  the one hand $a$, and on the other hand~$\neg b$ for every negative
  modal literal $\neg\Diamond b$ in~$\psi$.
  
\item In Presburger modal logic, let
  $\phi: = (\sharp(a)+\sharp(b)-\sharp(c)>0)$ (a conjunctive clause
  consisting of a single modal literal).  Then a one-step pair of the
  form $(\phi, \eta)$ is one-step satisfiable 
  iff~$\eta$ is consistent with
  $\rho:=(a\land b)\lor(a\land\neg c)\lor (b\land\neg c)$: For the `if'
  direction, note that~$\eta$ is consistent with some disjunct~$\rho'$
  of~$\rho$; we distinguish cases over~$\rho'$, and build a one-step
  model $(X,\tau,\mu)$ of $(\phi,\eta)$. In each case, we take~$X$ to
  consist of a single point~$1$; since~$\eta\land\rho'$ is consistent,
  we can pick~$\tau$ such that~$\tau(\eta\land\rho')=X$ (and hence
  $\tau(\eta)=X$). Moreover, we always take $\mu$ to be the multiset
  given by $\mu(1)=1$. If $\rho'=(a\land\neg c)$, then
  $\mu(\tau(a))+\mu(\tau(b))-\mu(\tau(c))=1+\mu(\tau(b))-0>0$, so
  $\mu\in\tau(\phi)$, and we are done. The case $\rho'=(b\land\neg c)$
  is analogous. Finally, if $\rho'=(a\land b)$, then
  $\mu(\tau(a))+\mu(\tau(b))-\mu(\tau(c))=2-\mu(\tau(c))>0$. For the
  `only if' direction, assume that $\eta\land\rho$ is inconsistent,
  so~$\eta$ propositionally entails $a\to c$, $b\to c$, and
  $\neg(a\land b)$, and let $(X,\tau,\mu)$ be a one-step model such
  that $\tau(\eta)=X$; we have to show that
  $\mu\notin\tau(\psi)$. Indeed, since $\tau(\eta)=X$ we have
  $\tau(a)\subseteq\tau(c)$, $\tau(b)\subseteq\tau(c)$, and
  $\tau(a)\cap\tau(b)=\emptyset$, so
  $\mu(\tau(a))+\mu(\tau(b))-\mu(\tau(c))\le 0$.
\item The reasoning in the previous example applies in the same way to
  one-step pairs of the form $(w(a)+w(b)-w(c)>0,\eta)$ in
  probabilistic modal logic. 
\item The example formula given in Remark~\ref{rem:graded-rules}
  translates into a one-step pair $(2\sharp(a)<1\land 2\sharp(a)>0,a)$
  in Presburger modal logic whose unsatisfiability does depend on
  multiplicities being integers; that is, the corresponding one-step
  pair $(2 w(a)<1\land 2 w(a)>0,a)$ in probabilistic modal logic is
  satisfiable.
\end{enumerate}
\end{example}
\begin{rem}
  For purposes of upper complexity bounds \PSpace and above for the
  strict one-step satisfiability problem, it does not matter whether
  the propositional component~$\eta$ of a one-step pair $(\psi,\eta)$
  is represented as a list or as a bit vector, as we have obvious
  mutual conversions between these formats that can be implemented
  using only polynomial space in~$|\Var(\psi)|$. For subexponential
  time bounds, on the other hand, the distinction between the formats
  does appear to matter, as the mentioned conversions do take
  exponential time in~$|\Var(\psi)|$.
\end{rem}
\noindent 
Note that most of a one-step pair $(\phi, \eta)$ is disregarded for
purposes of determining the input size of the \emph{strict} one-step
satisfiability problem, as~$\eta$ can be exponentially larger
than~$\phi$. Indeed, we have the following relationship between the
respective complexities of the lax one-step satisfiability problem and
the strict one-step satisfiability problem.
\begin{lemma}
  The strict one-step satisfiability problem of~$\Lambda$ is in
  \ExpTime iff the lax one-step satisfiability problem of~$\Lambda$
  can be solved on one-step pairs~$(\phi,\eta)$ in
  time~$2^{\CO((\log |\eta|+|\phi|)^k)}$ for some~$k$.
\end{lemma}
\noindent (Recent work on the coalgebraic $\mu$-calculus uses
essentially the second formulation~\cite{HausmannSchroder19}.)
\begin{proof}
  \emph{`Only if'} is trivial, since the time bound allows converting~$\eta$
  from the list representation assumed in the lax version of the
  problem to the bit vector representation assumed in the strict
  version. \emph{`If'}: Since we require that all variables
  mentioned by~$\eta$ occur also in~$\phi$, and assume that~$\eta$ is
  given in DNF, we have $|\eta|=2^{\CO(|\phi|)}$, so
  $\log|\eta|=\CO(|\phi|)$, and hence
  $2^{\CO((\log |\eta|+|\phi|)^k)}=2^{\CO(|\phi|^k)}$.
\end{proof}
\noindent We note that the one-step logic has an exponential-model
property (which in slightly disguised form has appeared first
as~\cite[Proposition 3.10]{SchroderPattinson06}):
\begin{lemma}\label{lem:one-step-models}
  A one-step pair $(\phi,\eta)$ over $V$ is satisfiable iff it is
  satisfiable by a one-step model of the form $(X,\tau,t)$ where $X$
  is the set of valuations $V\to 2$ satisfying $\eta$ (where
  $2=\{\top,\bot\}$ is the set of Booleans) and
  $\tau(a)=\{\kappa\in X\mid \kappa(a)=\top\}$ for $a\in V$.
\end{lemma}
\begin{proof}
  `If' is trivial; we prove `only if'.  Let $M=(Y,\theta,s)$ be a
one-step model of $(\phi,\eta)$. Take $X$ and $\tau$ as in the claim;
it is clear that $\tau(\eta)=X$. Define a map $f\colon  Y\to X$ by
$f(y)(a)=\top$ iff $y\in\theta(a)$ for $y\in Y$, $a\in V$. Then put
$t=Tf(s)\in TX$. By construction, we have $f^{-1}[\tau(a)]=\theta(a)$
for all $a\in V$. By naturality of predicate liftings and commutation
of preimage with Boolean operators, this implies that
$(Tf)^{-1}[\tau(\phi)]=\theta(\phi)$, so $s\in\theta(\phi)$ implies
$t=Tf(s)\in\tau(\phi)$; i.e.\ $(X,\tau,t)$ is a one-step model of
$(\phi,\eta)$.
\end{proof}
\noindent From the construction in the above lemma, we obtain the
following equivalent characterization of the one-step small model
property:
\begin{lemma}\label{lem:ospmp-log}
  The logic~$\Lambda$ has the (weak) one-step small model property iff there
  exists a polynomial~$p$ such that the following condition holds:
  Whenever a one-step pair $(\phi,\eta)$ is one-step satisfiable, then
  there exists $\eta'$ such that
  \begin{enumerate}
  \item $(\phi,\eta')$ is one-step satisfiable;
  \item the list representation of~$\eta'$ according to
    Definition~\ref{def:one-step} has size at most $p(|\Var(\phi)|)$
    (respectively at most $p(|\phi|)$); and
  \item $\eta'\PLentails\eta$.
  \end{enumerate}
\end{lemma}
\begin{proof}
  \emph{`Only if':} Take the conjunctive clauses of the DNF~$\eta'$ to
  be the ones realized in a polynomial-sized one-step model
  $(X,\tau,t)$ of $(\phi,\eta)$; that is,~$\eta'$ is the disjunction
  of all conjunctive clauses~$\rho$ mentioning all variables occurring
  in~$\phi$ such that $\tau(\rho)\neq\emptyset$.

  \emph{`If':} Take~$X$ as in Lemma~\ref{lem:one-step-models} and note
  that~$|X|$ is the number of conjunctive clauses in the
  representation of~$\eta'$ as per Definition~\ref{def:one-step}.
\end{proof}
\noindent Under the one-step small model property, the two versions of
the one-step satisfiability problem coincide for our purposes, as
detailed next. Recall that a multivalued function $f$ is
\emph{NPMV}~\cite{book-long-selman:npmv} if the representation length
of values of~$f$ on~$x$ is polynomially bounded in that of~$x$ and
moreover the graph of~$f$ is in~\NP; we generalize this notion
slightly to allow for size measures of~$x$ other than representation
length (such as the input size measure used in the strict one-step
satisfiability problem).  Most reasonable complexity classes
containing~\NP are closed under NPMV reductions; in particular this
holds for \PSpace, \ExpTime, and all levels of the polynomial
hierarchy.
\begin{lemma}\label{lem:oss}
  Let $\Lambda$ have the weak one-step small model property
  (Definition~\ref{def:one-step}). Then the strict one-step
  satisfiability problem of~$\Lambda$ is NPMV-reducible to lax
  one-step satisfiability. In particular, if lax one-step
  satisfiability is in \NP (\PSpace/\ExpTime), then strict one-step
  satisfiability is in \NP (\PSpace/\ExpTime).
\end{lemma}
\begin{proof}
  By Lemma~\ref{lem:ospmp-log}, and in the notation of its statement,
  the NPMV function that maps $(\phi,\eta)$ (with~$\eta$ in bit vector
  representation) to all $(\phi,\eta')$ with $\eta'$ of (list)
  representation size at most $p(|\phi|)$ and $\eta'\PLentails\eta$
  reduces strict one-step satisfiability to lax one-step
  satisfiability.
\end{proof}
\noindent Of the two versions of the one-step small model property,
the stronger version (polynomial in $|\Var(\phi)|$) turns out to be
prevalent in the examples. The weak version (polynomial in $|\phi|$)
is of interest mainly due to the following equivalent
characterization:
\begin{theorem}
  Suppose that the lax one-step satisfiability problem of~$\Lambda$ is
  in \NP. Then the weak one-step small model property holds
  for~$\Lambda$ iff the strict one-step satisfiability problem
  of~$\Lambda$ is in \NP.
\end{theorem}
\begin{proof}
  `Only if' is immediate by Lemma~\ref{lem:oss}; we prove
  `if'. Let~$\mathsf{M}$ be a non-deterministic (random access) Turing machine
  that solves the strict one-step satisfiability problem in polynomial
  time, and let the one-step pair $(\phi,\eta)$ be one-step
  satisfiable. Then~$\mathsf{M}$ has a successful run on~$(\phi,\eta)$. Since
  this run takes polynomial time in~$|\phi|$, it accesses only
  polynomially many bits in the bit vector representation
  of~$\eta$. We can therefore set all other bits to~$0$, obtaining a
  polynomial-sized DNF~$\eta'$ such that $\eta'\PLentails\eta$ and
  $(\phi,\eta')$ is still one-step satisfiable, as witnessed by
  otherwise the same run of~$\mathsf{M}$. By Lemma~\ref{lem:ospmp-log}, this
  proves the weak one-step small model property.
\end{proof}
\noindent Although not phrased in these terms, the complexity analysis
of Presburger modal logic (without global assumptions) by Demri and
Lugiez~\cite{DemriLugiez10} is based on showing that the strict
one-step satisfiability problem is in
\PSpace~\cite{SchroderPattinson08d}, without using the one-step small
model property for Presburger modal logic -- in fact, our proof of the
latter is based on more recent results from integer programming: We
recall that the classical \emph{Carath\'eodory theorem}
(e.g.~\cite{Schrijver86}) may be phrased as saying that every system
of~$d$ linear equations that has a solution over the non-negative
reals has such a solution with at most~$d$ non-zero
components. Eisenbrand and Shmonin~\cite{EisenbrandShmonin06} prove an
analogue over the integers, which we correspondingly rephrase as
follows.
\begin{lemma}[Integer Carath\'eodory theorem
  \cite{EisenbrandShmonin06}] \label{lem:eisenbrand} Every system
  of~$d$ linear equations $\sum u_i x_i=\inhom$ with integer
  coefficients~$u_i$ of binary length at most~$s$ that has a solution
  over the non-negative integers has such a solution with at most
  polynomially many non-zero components in~$d$ and~$s$ (specifically,
  $\CO(sd\log d)$).
\end{lemma}
\noindent To deal with lax one-step satisfiability, we will moreover
need the well-known result by Papadimitriou that establishes a
polynomial bound on the size of components of solutions of systems of
integer linear equations:
\begin{lemma}\label{lem:papadimitriou}\cite{Papadimitriou81}
  Every system of integer linear
  equations 
  in variables $x_1,\dots,x_n$ that has a solution over the
  non-negative integers has such a solution
  $(\hat x_1,\dots,\hat x_n)$ with the binary length of each
  component~$\hat x_i$ bounded polynomially in the overall binary
  representation size of the equation system.
\end{lemma}
\begin{corollary}\label{cor:int-constr}
  Solvability of systems of Presburger constraints is in~\NP.
\end{corollary}
\begin{proof}
  It suffices to show that we can generalize
  Lemma~\ref{lem:papadimitriou} to systems of Presburger constraints.
  Indeed, we can reduce Presburger constraints to equations involving
  additional variables. Specifically, we replace an inequality
  $\sum u_i \cdot x_i > \inhom$ with the equation
  $\sum u_i \cdot x_i - y = \inhom + 1$ and a modular constraint
  $\sum u_i \cdot x_i \eqmod{k} \inhom$ with either
  $\sum u_i \cdot x_i -k\cdot y = \inhom$ or $\sum u_i \cdot x_i + k\cdot y = \inhom$,
  depending on whether the given solution satisfies
  $\sum u_i \cdot x_i \ge \inhom$ or $\sum u_i \cdot x_i \le \inhom$; in every
  such replacement, choose~$y$ as a fresh variable.
\end{proof}
%
\noindent From these observations, we obtain sufficient tractability
of strict one-step satisfiability in our key examples:
\begin{example}\label{expl:ossmp}
  \begin{enumerate}[wide]
  \item \label{expl:osmp-presburger}\emph{Presburger modal logic has
      the one-step small model property}. To see this, let a one-step
    pair $(\phi,\eta)$ over $V=\{a_1,\dots,a_n\}$ be satisfied by a
    one-step model $M=(X,\tau,\mu)$, where by
    Lemma~\ref{lem:one-step-models} we can assume that $X$ consists of
    satisfying valuations of $\eta$, hence has at most exponential
    size in~$|\phi|$. Put
    $q_i=\mu(\tau(a_i))$. 
    Now all we need to know about~$\mu$ to guarantee that~$M$
    satisfies $\phi$ is that the (non-negative integer) numbers
    $y_x:=\mu(x)$, for $x\in X$, satisfy
    \begin{equation*}\textstyle
      \sum_{x\in\tau(a_i)}y_x=q_i\qquad\text{for $i=1,\dots,n$}.
    \end{equation*}
    We can see this as a system of~$n$ linear equations in the
    $y_x$, which by the integer Carath\'eodory theorem
    (Lemma~\ref{lem:eisenbrand}) has a non-negative integer solution
    $(y'_x)_{x\in X}$ with only~$m$ nonzero components where~$m$ is
    polynomially bounded in~$n$ (the coefficients of the~$y_x$ all
    being~$1$), and hence in~$|\phi|$; from this solution, we
    immediately obtain a one-step model $(X',\tau',\mu')$ of
    $(\phi,\eta)$ with $m$ states. Specifically, take
    $X'=\{x\in X\mid y'_x>0\}$, $\tau'(a_i)=\tau(a_i)\cap X'$ for
    $i=1,\dots,n$, and $\mu'(x)=y'_x$ for~$x\in X'$.

    Moreover, again using Lemma~\ref{lem:one-step-models}, lax
    one-step satisfiability in Presburger modal logic reduces
    straightforwardly to checking solvability of Presburger
    constraints over the non-negative integers, which by
    Corollary~\ref{cor:int-constr} can be done in~NP. Specifically,
    given a one-step pair $(\phi,\eta)$, with~$\eta$ represented as
    per Definition~\ref{def:one-step}, introduce a variable~$x_{\rho}$
    for every conjunctive clause~$\rho$ of~$\eta$ (i.e.\ for every
    valuation satisfying~$\eta$), and translate every constraint
    $\sum_i u_i\cdot\sharp(a_i)\sim v$ in~$\phi$ into
    \begin{equation*}
      \sum_i u_i\cdot\sum_{\rho\PLentails\eta\land a_i}x_\rho\sim v.
    \end{equation*}
    Thus, the lax one-step satisfiability problem of Presburger modal logic is in \NP,
    and by Lemma~\ref{lem:oss}, we obtain that \emph{strict one-step
      satisfiability in Presburger modal logic is in \NP}.
  \item By a completely analogous argument as for Presburger modal
    logic (using the standard Carath\'eodory theorem),
    \emph{probabilistic modal logic with polynomial inequalities has
      the one-step small model property}. Moreover, lax one-step
    satisfiability reduces, analogously as in the previous item, to
    solvability of systems of polynomial inequalities over the reals,
    which can be checked in \PSpace~\cite{Canny88} (this argument can
    essentially be found in~\cite{FaginHalpernMegiddo90}). Again, we
    obtain that \emph{strict one-step satisfiability in probabilistic
      modal logic with polynomial inequalities is in \PSpace}.
  \end{enumerate}
\end{example}
  
\begin{rem}[Variants of the running examples]
  The proof of the one-step small model property for Presburger modal
  logic and probabilistic modal logic with polynomial inequalities
  will in both cases work for any modal logic over integer- or
  real-weighted systems, respectively, whose modalities depend only on
  the measures of their arguments; call such modalities \emph{fully
    explicit}. There are quite sensible operators that violate this
  restriction; e.g.\ an operator $I(\phi,\psi)$ `$\phi$ is independent
  of~$\psi$' would depend on the probabilities of $\phi$ and~$\psi$
  but also on that of $\phi\land\psi$. Indeed, in this vein we easily
  obtain a natural logic over probabilistic systems that fails to have
  the one-step small model property: If we generalize the independence
  modality~$I$ to several arguments and combine it with operators
  $w(-)>0$ stating that their arguments have positive probability,
  then every one-step model of the one-step pair
  \begin{equation*}\textstyle
    (I(a_1,\dots,a_n)\land\Land_{i=1}^nw(a_i)>0 \land\Land_{i=1}^nw(\neg a_i)>0,\top)
  \end{equation*}
  has at least $2^n$ states.

  However, a completely analogous argument as in the proof of
  Lemma~\ref{lem:one-step-models} shows that every predicate lifting
  for functors such as $\Dist$, $\SDist$, or~$\Bag$ depends only on
  the measures of Boolean combinations of its arguments, which can
  equally well be expressed using the propositional operators of the
  logic. That is, every coalgebraic modal logic over weighted systems
  translates (possibly with exponential blowup) into one that has only
  fully explicit modalities and hence has the one-step small model
  property, as exemplified for the case of~$I$ in
  Section~\ref{sec:prob}.

  Incidentally, a similar example as the above produces a natural
  example of a logic that does not have the one-step small model
  property but whose lax one-step satisfiability problem is
  nevertheless in~\ExpTime. Consider a variant of probabilistic modal
  logic (Section~\ref{sec:prob}) featuring \emph{linear} (rather than
  polynomial) inequalities over probabilities $w(\phi)$, and
  additionally \emph{fixed-probability conditional independence}
  operators $I_{p_1,\dots,p_n}$ of arity~$n+1$ for $n\ge 1$ and
  $p_1,\dots,p_n\in\Rat\cap[0,1]$. The application of
  $I_{p_1,\dots,p_n}$ to formulae $\phi_1,\dots,\phi_n,\psi$ is
  written $I_{p_1,\dots,p_n}(\phi_1,\dots,\phi_n\mid \psi)$, and read
  `$\phi_1,\dots,\phi_n$ are conditionally independent given~$\psi$,
  and each~$\phi_i$ has conditional probability $p_i$ given~$\psi$'.
  A one-step modal literal
  $I_{p_1,\dots,p_n}(a_1,\dots,a_n|b)$ translates, by definition, into linear equalities
  \begin{equation*}\textstyle
    w(\Land_{i\in I}a_i)
    -(\prod_{i\in I}p_i) w(\psi)=0\qquad\text{for all $I\subseteq\{1,\dots,n\}$.}
  \end{equation*}
  Thus, a given one-step clause $\psi$ generates, in the same way as
  previously, a system of linear inequalities, now of exponential size
  in~$|\psi|$. Since solvability of systems of linear inequalities can,
  by standard results in linear programming~\cite{Schrijver86}, be
  checked in polynomial time, we obtain that the strict one-step
  satisfiability problem is in \ExpTime as claimed. On the other hand,
  the one-step small model property fails for the same reasons as
  for the~$I$ operator described above.
\end{rem}
\noindent By previous results in coalgebraic
logic~\cite{SchroderPattinson08d}, the observations in
Example~\ref{expl:ossmp}.\ref{expl:osmp-presburger} imply decidability
in \PSpace of the respective \emph{plain} satisfiability problems,
reproducing a previous result by Demri and Lugiez~\cite{DemriLugiez10}
for the case of Presburger modal logic; we show in
Section~\ref{sec:type-elim} that the same observations yield an
optimal upper bound \ExpTime for satisfiability under global
assumptions.

\begin{rem}[Comparison with tractable modal rule sets]\label{rem:rules}
  Most previous generic complexity results in coalgebraic logic have
  relied on complete sets of modal tableau rules that are sufficiently
  tractable for purposes of the respective complexity bound,
  e.g.~\cite{SchroderPattinson09a,SchroderEA09,GoreEA10a}. We briefly
  discuss how these assumptions imply the ones used in the present
  paper.

  The rules in question (\emph{one-step tableau rules}) are of the
  shape $\phi/\rho$ where $\phi$ is a modal conjunctive clause over
  $V$ and $\rho\in\Prop(V)$, subject to the same syntactic
  restrictions as one-step pairs, i.e.~$\phi$ must be clean and~$\rho$
  can only mention variables occurring in~$\phi$.  Such rules form
  part of a tableau system that includes also the standard
  propositional rules. As usual in tableau systems, algorithms for
  satisfiability checking based on the tableau rules proceed roughly
  according to the principle `in order to establish that~$\psi$ is
  satisfiable, show that the conclusions of all rule matches to~$\psi$
  are satisfiable' (this is dual to validity checking via formal proof
  rules, where to show that~$\psi$ is valid one needs to find some
  proof rule whose conclusion matches~$\psi$ and whose premiss is
  valid). More precisely, the (one-step) soundness and completeness
  requirement on a rule set~$\Rules$ demands that a one-step
  pair~$(\psi,\eta)$ is satisfiable iff for every rule $\phi/\rho$
  in~$\Rules$ and every injective variable renaming~$\sigma$ such that
  $\psi\PLentails\phi\sigma$ (see Definition~\ref{def:prop} for the
  notation~$\PLentails$), the propositional formula
  $\eta\land\rho\sigma$ is satisfiable. Since~$\psi$ and~$\phi$ are
  modal conjunctive clauses (and~$\psi$, being clean, cannot contain
  clashing modal literals), $\psi\PLentails\phi\sigma$ means that
  $\psi$ contains every modal literal of $\phi\sigma$.

  The exact requirements on tractability of a rule set vary with the
  intended complexity bound for the full logic. In connection with
  \ExpTime bounds, one uses \emph{exponential tractability} of the
  rule set (e.g.~\cite{CirsteaEA11}). This condition requires that
  rules have an encoding as strings such that every rule $\phi/\rho$
  in~$\Rules$ that \emph{matches} a given modal conjunctive
  clause~$\psi$ over~$V$ \emph{under} a given injective
  renaming~$\sigma$, i.e.\ $\psi\PLentails\phi\sigma$, has an encoding
  of polynomial size in~$\psi$, and moreover given a modal conjunctive
  clause~$\psi$ over~$V$, it can be decided in exponential time
  in~$|\psi|$ whether (i) an encoded rule $\phi/\rho$ matches~$\psi$
  under a given renaming~$\sigma$, and (ii) whether a given
  conjunctive clause~$\chi$ over~$\Var(\psi)$ propositionally entails
  the conclusion~$\rho\sigma$ the instance $\phi\sigma/\rho\sigma$ of
  an encoded rule $\phi/\rho$ under a given renaming~$\sigma$.

  Now suppose that a set~$\Rules$ of modal tableau rules satisfies all
  these requirements, i.e.\ is one-step sound and complete for the
  given logic and exponentially tractable, with polynomial bound~$p$
  on the size of rule codes. Then one sees easily that the strict
  one-step satisfiability problem is in \ExpTime: Given a one-step
  pair $(\psi,\eta)$ to be checked for one-step satisfiability, we can
  go through all rules $\phi/\rho$ represented by codes of length at
  most $p(|\psi|)$ and all injective renamings~$\sigma$ of the
  variables of~$\phi$ into the variables of~$\psi$ such that
  $\phi/\rho$ matches~$\psi$ under~$\sigma$, and then for each such
  match go through all conjunctive clauses~$\chi$ over $\Var(\psi)$
  that propositionally entail~$\rho\sigma$, checking for each
  such~$\chi$ that $\eta\land\chi$ is propositionally
  satisfiable. Both loops go through exponentially many iterations,
  and all computations involved take at most exponential time.
  Summing up, complexity bounds obtained by our current semantic
  approach subsume earlier tableau-based ones.
\end{rem}
\section{Type Elimination}\label{sec:type-elim} 



\noindent We now describe a type elimination algorithm that realizes
an \ExpTime upper bound for reasoning with global assumptions in
coalgebraic logics. Like all type elimination algorithms, it is not
suited for practical use, as it begins by constructing the full
exponential-sized set of types (in the initialization phase of the
computation of a greatest fixpoint). We therefore refine the algorithm
to a global caching algorithm in Section~\ref{sec:caching}.

As usual, we rely on defining a scope of relevant formulae:
\begin{definition}
  We define \emph{normalized negation} $\nneg$ by taking
  $\nneg\phi=\phi'$ if a formula $\phi$ has the form $\phi=\neg\phi'$,
  and $\nneg\phi=\neg\phi$ otherwise.  A set $\Sigma$ of formulae is
  \emph{closed} if $\Sigma$ is closed under subformulae and normalized
  negation. The \emph{closure} of a set~$\Gamma$ of formulae is the
  least closed set containing~$\Gamma$.
\end{definition}
\noindent We \emph{fix from now on a global assumption $\psi$ and a
  formula $\phi_0$ to be checked for $\psi$-satisfiability}. We denote
the closure of $\{\psi,\phi_0\}$ in the above sense by $\Sigma$. Next,
we approximate the $\psi$-satisfiable subsets of $\Sigma$ from above
via a notion of type that takes into account only propositional
reasoning and the global assumption~$\psi$:
\begin{definition}\label{def:type}
  A \emph{$\psi$-type} is a subset $\type\subseteq\Sigma$ such that
  \begin{itemize}
  \item $\psi\in \type\not\owns\bot$;
  \item whenever $\neg \phi\in\Sigma$, then $\neg \phi\in \type$ iff
    $\phi\notin \type$;
  \item whenever $\phi\land\chi\in\Sigma$, then
    $\phi \land \chi \in \type$ iff $\phi,\chi\in \type$.
  \end{itemize}
\end{definition}
\noindent
The design of the algorithm relies on one-step satisfiability as an
abstraction: We denote the set of all $\psi$-types by $\types{\psi}$. For a
formula $\phi\in\Sigma$, we put
\begin{equation*}
  \hat \phi=\{\type\in \types{\psi}\mid \phi\in \type\},
\end{equation*}
intending to construct a model on a suitable subset
$S\subseteq\types{\psi}$ in such a way that $\hat\phi\cap S$ becomes
the extension of~$\phi$. We take $\Sigmavars$ to be the set of
propositional variables $a_{\hearts\rho}$ for all modal atoms
$\hearts\rho\in\Sigma$; we then define a substitution $\Sigmasubst$ by
$\Sigmasubst(a_{\hearts\rho})=\rho$ for
$a_{\hearts\rho}\in \Sigmavars$. For $S\subseteq \types{\psi}$ and
$\type\in S$, we construct a one-step pair
\begin{equation*}
  (\phi_\type,\eta_S)
\end{equation*}
over $\Sigmavars$ by taking $\phi_\type$ to be the conjunction of all
modal literals $\epsilon\hearts a_{\hearts\rho}$ over $\Sigmavars$
such that $ \epsilon\hearts\rho\in \type$ (note that indexing the
propositional variables~$a_{\hearts\rho}$ over $\hearts\rho$ instead
of just~$\rho$ ensures that~$\psi_\Gamma$ is clean as required), and
$\eta_S$ to be the DNF (for definiteness, in bit vector representation
as per Definition~\ref{def:one-step}) containing for each
$\typeb\in S$ a conjunctive clause
\begin{equation*}
  \Land_{\hearts\rho\in\Sigma\mid\rho\in\typeb}a_{\hearts\rho}\land
  \Land_{\hearts\rho\in\Sigma\mid\nneg\rho\in\typeb}\neg a_{\hearts\rho}.  
\end{equation*}
That is,~$\phi_\Gamma$ arises from~$\Gamma$ by abstracting the
arguments~$\rho$ of modalized formulae~$\hearts\rho\in\Gamma$ as
propositional variables~$a_{\hearts\rho}$, and~$\eta$ captures the
propositional dependencies that will hold in~$S$ among these arguments
if the construction works as intended. We define a functional
\begin{equation}\label{eq:elim-functional}
 \begin{array}{lcll}
   \CE\colon &\Pow(\types{\psi})&\to & \Pow(\types{\psi})\\[0.3ex]
     & S & \mapsto & \{\type \in S \mid (\phi_\type,\eta_S)\text{ is one-step satisfiable}\},
 \end{array}
\end{equation}
whose greatest fixpoint $\nu\CE$ will turn out to contain precisely
the satisfiable types. Existence of $\nu\CE$ is guaranteed by the
Knaster-Tarski fixpoint theorem and the following lemma:
\begin{lemma}
  The functional $\CE$ is monotone w.r.t.\ set inclusion.
\end{lemma}
\begin{proof}
  For $S\subseteq S'$, the DNF $\eta_{S'}$ is weaker
  than~$\eta_S$, as it contains more disjuncts.
\end{proof}
\noindent By Kleene's fixpoint theorem, we can compute $\nu\CE$ by
just iterating $\CE$:
\begin{alg}\label{alg:type-elim}
  (Decide by type elimination whether $\phi_0$ is satisfiable over $\psi$)
  \begin{enumerate}
  \item Set $S:=\types{\psi}$.
  \item Compute $S'=\CE(S)$; if $S'\neq S$ then put $S:=S'$ and repeat.
  \item Return `yes' if $\phi_0\in \type$ for some $\type\in S$, and `no'
    otherwise.
  \end{enumerate}
\end{alg}
\noindent The run time analysis is straightforward:
\begin{lemma}\label{lem:type-elim-time}
  \noindent If the strict one-step satisfiability problem of~$\Lambda$
  is in \ExpTime, then Algorithm~\ref{alg:type-elim} has at most
  exponential run time.
\end{lemma}
\begin{proof}
  Since~$\types{\psi}$ has at most exponential size, the algorithm
  runs through at most exponentially many iterations. In a single
  iteration, we have to compute $\CE(S)$, checking for each of the at
  most exponentially many $\type\in S$ whether~$(\phi_\type,\eta_S)$
  is one-step satisfiable. The assumption of the lemma guarantees that
  each one-step satisfiability check takes only exponential time, as
  $\phi_\type$ is of linear size.
\end{proof}
\noindent It remains to prove correctness of the algorithm; that is,
we show that, as announced above, $\nu\CE$ consists precisely of the
$\psi$-satisfiable types. We split this claim into two inclusions,
corresponding to soundness and completeness, respectively:
\begin{lemma}\label{lem:realization}
  The set of $\psi$-satisfiable types is a postfixpoint of $\CE$.
\end{lemma}
\noindent (Since $\nu\CE$ is also the greatest postfixpoint of~$\CE$,
this implies that $\nu\CE$ contains all $\psi$-satisfiable types. This
means that Algorithm~\ref{alg:type-elim} is \emph{sound}, i.e.\
answers `yes' on $\psi$-satisfiable formulae.)
\begin{proof}
  Let $R$ be the set of $\psi$-satisfiable types; we have to show that
  $R \subseteq \CE(R)$.  So let $\type\in R$; then we have a state~$x$
  in a $\psi$-model $C = (X, \gamma)$ such that $x\models_C\type$. By
  definition of~$\CE$, we have to show that the one-step pair
  $(\phi_\type, \eta_R)$ is one-step satisfiable. We claim that the
  one-step model $M=(X,\tau,\xi(x))$, where~$\tau$ is defined by
  \begin{equation*}
    \tau(a_{\hearts\rho}):=\Sem{\Sigmasubst(a_{\hearts\rho})}_C=\Sem{\rho}_C
  \end{equation*}
  for $a_{\hearts\rho}\in \Sigmavars$, satisfies
  $(\phi_\type,\eta_R)$. For the propositional part~$\eta_R$, let
  $y\in X$; we have to show $y\in\tau(\eta_R)$. Put
  $\Delta=\{\rho\in\Sigma\mid\ y\models\rho\}$. Then $\Delta\in R$, so
  that $\eta_R$ contains the conjunctive clause
  \begin{equation*}
    \theta:=\Land_{\hearts\rho\in\Sigma\mid\rho\in\Delta}a_{\hearts\rho}\land\Land_{\hearts\rho\in\Sigma\mid\rho\notin\Delta}\neg
    a_{\hearts\rho}.
  \end{equation*}
  By the definitions of~$\tau$ and~$\theta$, we have
  $y\in\tau(\theta)\subseteq\tau(\eta_R)$, as required (e.g.~if
  $\hearts\rho\in\Sigma$ and $\rho\in\Delta$, then $y\models\rho$,
  i.e.~$y\in\Sem{\rho}_C=\tau(a_{\hearts\rho})$; the negative case is
  similar). Finally, for $\psi_\type$, let $\hearts\rho\in\Sigma$; we
  have to show that $\hearts\rho\in\type$ iff
  $\xi(x)\in\Sem{\hearts}(\tau(a_{\hearts\rho}))=\Sem{\hearts}(\Sem{\rho}_C)$. But
  the latter just means that $x\models\hearts\rho$, so the equivalence
  holds because~$x\models\Gamma$.
\end{proof}
\noindent For the converse inclusion, i.e.~completeness, we show the
following (combining the usual existence and truth lemmas):
\begin{lemma}\label{lem:ex-truth}
  Let $S$ be a postfixpoint of $\CE$. Then there exists a
  $T$-coalgebra $C=(S,\gamma)$ such that for each $\rho\in\Sigma$,
  $\Sem{\rho}_C=\hat\rho\cap S$. 
\end{lemma}
\begin{proof}
  To construct the transition structure $\gamma$, let $\type\in
  S$. Since~$S$ is a postfixpoint of~$\CE$, the one-step pair
  $(\phi_\type,\eta_S)$ is satisfiable; let $(X,\tau,t)$ be a one-step
  model of $(\phi_\type,\eta_S)$.
  By 
  construction of~$\eta_S$, we then have a map $f:X\to S$ such that for all
  $\hearts\rho\in\Sigma$,
  \begin{equation}\label{eq:def-f}
    x\in\tau(a_{\hearts\rho})\quad\text{iff}\quad\rho\in f(x)\quad\text{iff}\quad
    f(x)\in\hat\rho.
  \end{equation}
  We put $\gamma(\type)=Tf(t)\in TS$.  For the $T$-coalgebra
  $C=(S,\gamma)$ thus obtained, we show the claim
  $\Sem{\rho}_C=\hat \rho\cap S$ by induction over $\rho\in\Sigma$.
  The propositional cases are by the defining properties of types
  (Definition~\ref{def:type}). For the modal case, we have (for
  $\type$ and associated data $f,t$ as above)
  \begin{align*}
    \type\models\hearts\rho & \iff \gamma(\type)=Tf(t)\in\Sem{\hearts}_S(\Sem{\rho}_C)\\
                            & \iff  t\in\Sem{\hearts}_X(f^{-1}[\Sem{\rho}_C]) &&\by{naturality}\\
                            & \qquad\qquad = \Sem{\hearts}_X(f^{-1}[\hat\rho\cap S]) && \by{induction}\\
                            & \qquad\qquad = \Sem{\hearts}_X(\tau(a_{\hearts\rho})) &&\by{\ref{eq:def-f}}\\
    & \iff \hearts\rho\in\type && \by{definition of~$\phi_\type$} \qedhere
  \end{align*}
\end{proof}
\noindent A $T$-coalgebra as in Lemma~\ref{lem:ex-truth} is clearly a
$\psi$-model, so the above lemma implies that every postfixpoint
of~$\CE$, including~$\nu\CE$, consists only of $\psi$-satisfiable
types. That is, that Algorithm~\ref{alg:type-elim} is indeed complete,
i.e.\ answers `yes' \emph{only} on $\psi$-satisfiable formulae. This
completes the correctness proof of Algorithm~\ref{alg:type-elim}; in
combination with the run time analysis
(Lemma~\ref{lem:type-elim-time}) we thus obtain
\begin{theorem}[Complexity of satisfiability under global assumptions]\label{thm:exptime}
  If the strict one-step satisfiability problem of the logic~$\Lambda$
  is in \ExpTime, then satisfiability under global assumptions
  in~$\Lambda$ is in \ExpTime.
\end{theorem}
\begin{example}
  By the results of the previous section (Example~\ref{expl:ossmp})
  and by inheriting lower bounds from reasoning with global
  assumptions in $K$~\cite{FischerLadner79}, we obtain that reasoning
  with global assumptions in Presburger modal logic and in
  probabilistic modal logic with polynomial inequalities is
  \ExpTime-complete. We note additionally that the same holds also for
  our separating example, probabilistic modal logic with linear
  inequalities and fixed-probability independence operators (which
  does not have the one-step small model property but whose strict
  one-step satisfiability problem is nevertheless in \ExpTime).
\end{example}

\section{Global Caching}
\label{sec:caching}

\noindent We now develop the type elimination algorithm from the
preceding section into a global caching algorithm. Roughly speaking,
global caching algorithms perform \emph{expansion} steps, in which new
nodes to be explored are added to the tableau, and \emph{propagation}
steps, in which the satisfiability (or unsatisfiability) is determined
for those nodes for which the tableau already contains enough
information to allow this. The practical efficiency of global caching
algorithms is based on the fact that the algorithm can stop as soon as
the root node is marked satisfiable or unsatisfiable in a propagation
step, thus potentially avoiding generation of all (exponentially many)
possible nodes.  Existing global caching algorithms work with systems
of tableau rules (satisfiability is guaranteed if every applicable
rule has at least one satisfiable conclusion)~\cite{GoreEA10a}. The
fact that we work with a semantics-based decision procedure impacts on
the design of the algorithm in two ways:
\begin{itemize}
\item In a tableaux setting, node generation in the expansion steps is
  driven by the tableau rules, and a global caching algorithm
  generates modal successor nodes by applying tableau rules. In
  principle, however, modal successor nodes can be generated at will,
  with the rules just pointing to relevant
  nodes. 
  In our setting, we
  make the relevant nodes explicit using the concept of
  \emph{children}.
\item The rules govern the propagation of satisfiability and
  unsatisfiability among the nodes. Semantic propagation of
  satisfiability is straightforward, but propagation of
  unsatisfiability again needs the concept of children: a (modal) node
  can only be marked as unsatisfiable once all its children have been
  generated (and too many of them are unsatisfiable).
\end{itemize}
\noindent We continue to work with a closed set $\Sigma$ as in
Section~\ref{sec:type-elim} (generated by the global assumption $\psi$
and the target formula $\phi_0$) but replace types with
\emph{(tableau) sequents}, i.e.\ arbitrary subsets
$\Gamma,\Theta\subseteq\Sigma$, understood conjunctively; in
particular, a sequent need not determine the truth of every formula
in~$\Sigma$. We write $\Seqs=\Pow\Sigma$, and occasionally refer to
sequents as \emph{nodes} in allusion to an implicit graphical
structure (made more explicit in Section~\ref{sec:concrete-alg}). A
\emph{state} is a sequent consisting of modal literals only (recall
that we regard propositional atoms as nullary modalities; so if
propositional atoms in this sense are part of the logic, then states
may also contain propositional atoms or their negations).  We denote
the set of states by $\States$.

To convert sequents into states, we employ the usual
\emph{propositional rules}
\begin{equation*}
  \infrule{\Gamma,\phi_1\land \phi_2}{\Gamma,\phi_1,\phi_2}
  \quad
  \infrule{\Gamma,\neg(\phi_1\land \phi_2)}{\Gamma,\neg \phi_1\mid\Gamma,\neg \phi_2}
  \quad
  \infrule{\Gamma,\neg\neg \phi}{\Gamma,\phi}
  \quad 
  \infrule{\Gamma,\bot}{}
\end{equation*}
where $\mid$ separates alternative conclusions (and the last rule has
no conclusion).
\begin{rem}
  Completeness of the global caching algorithm will imply that the
  usual clash rule $\Gamma,\phi,\neg\phi/\;$ (a rule with no
  conclusions, like the rule for~$\bot$ above) is admissible. Notice
  that in logics featuring propositional atoms~$p$, i.e.\ nullary
  modalities, the atomic clash rule $\Gamma,p,\neg p/$ would be
  considered a modal rule.
\end{rem}

\noindent As indicated above, the expansion steps of the algorithm
will be driven by the following child relation on tableau sequents:
\begin{definition}
  The \emph{children} of a state $\Gamma$ are the sequents consisting
  of~$\psi$ and, for each modal literal
  $\epsilon\hearts\phi\in\Gamma$, a choice of either $\phi$ or
  $\neg \phi$. The \emph{children} of a non-state sequent are its
  conclusions under the propositional rules. In both cases, we write
  $\children{\Gamma}$ for the set of children of~$\Gamma$.
\end{definition}
\noindent For purposes of the global caching algorithm, we modify the
functional~$\CE$ defined in Section~\ref{sec:type-elim} to work also
with sequents (rather than only types) and to depend on a set
$\Nodes\subseteq\Seqs$ of sequents already generated. To this end, we
introduce for each state~$\Gamma\in\Nodes$ a set~$V_\Gamma$ containing
a propositional variable $a_{\epsilon\hearts\rho}$ for each modal
literal $\epsilon\hearts\rho\in\Gamma$, as well as a substitution
$\sigma_\Gamma$ on~$V_\Gamma$ defined by
$\sigma_\Gamma(a_{\epsilon\hearts\rho})=\rho$. Given
$S\subseteq\Nodes$, we then define a one-step pair
$(\phi_\Gamma,\eta_S)$ over~$V_\Gamma$ similarly as in
Section~\ref{sec:type-elim}: We take~$\phi_\type$ to be the
conjunction of all modal literals
$\epsilon\hearts a_{\epsilon\hearts\rho}$ over $V_\Gamma$ such that
$ \epsilon\hearts\sigma_\Gamma(a_{\epsilon\hearts\rho})=
\epsilon\hearts\rho\in \type$ (we need to
index~$a_{\epsilon\hearts\rho}$ over $\epsilon\hearts\rho$ instead of
just $\hearts\rho$ to ensure that $\phi_\Gamma$ is clean, since
sequents, unlike types, may contain clashes), and $\eta_S$ to be the
DNF containing for each $\typeb\in S$ a conjunctive clause
\begin{equation*}
  \Land_{\epsilon\hearts\rho\in\Gamma\mid \rho\in\typeb}a_{\epsilon\hearts\rho}\land
  \Land_{\epsilon\hearts\rho\in\Gamma\mid\nneg\rho\in\typeb}\neg a_{\hearts\rho}.  
\end{equation*}
We now define a functional
\begin{equation*}
\CE_{\Nodes}\colon \Pow\Nodes\to\Pow\Nodes
\end{equation*}
by taking $\CE_{\Nodes}(S)$ to consist of
\begin{itemize}
\item all non-state sequents $\Gamma\in \Nodes\setminus\States$ such
  that $S\cap\children{\Gamma} \neq\emptyset$ (i.e.\ some
  propositional rule that applies to $\Gamma$ has a conclusion that is
  contained in~$S$), and
\item all states $\Gamma\in \Nodes\cap\States$ such that the one-step pair
  $(\phi_\Gamma,\eta_{S\cap\children{\Gamma}})$ is one-step satisfiable.
\end{itemize}
To propagate unsatisfiability, we introduce a second functional
$\CA_{\Nodes}\colon \Pow\Nodes\to\Pow\Nodes$,  where we take
$\CA_{\Nodes}(S)$ to consist of
\begin{itemize}
\item all non-state sequents $\Gamma\in \Nodes\setminus\States$ such
  that there is a propositional rule applying to $\Gamma$ all whose
  conclusions are in $S$, and
\item all states $\Gamma\in \Nodes\cap\States$ such that
  $\children{\Gamma} \subseteq G$ and the one-step pair
  $(\phi_\Gamma,\eta_{\children{\Gamma}\setminus S})$ is one-step
  unsatisfiable.
\end{itemize}
\noindent Both~$\CE_G$ and~$\CA_G$ are clearly monotone.  We note
additionally that they also depend monotonically on~$G$:
\begin{lemma}\label{lem:functionals-monotone}
  Let $G\subseteq G'\subseteq\Seqs$. Then
  \begin{enumerate}
  \item\label{item:functionals-monotone} $\CE_G(S)\subseteq\CE_{G'}(S)$ and
    $\CA_G(S)\subseteq\CA_{G'}(S)$ for all~$S\in\Pow G$; 
  \item\label{item:fps-monotone} $\nu\CE_G\subseteq\nu\CE_{G'}$ and
    $\mu\CA_G\subseteq\mu\CA_{G'}$.
\end{enumerate}
\end{lemma}
\begin{proof}
  Claim~\eqref{item:functionals-monotone} is immediate from the
  definitions (for~$\CA_G$, this hinges on the condition
  $\children{\Gamma}\subseteq G$ for states~$\Gamma$); we show
  Claim~\eqref{item:fps-monotone}. For $\CE_G$, it suffices to show
  that $\nu\CE_G$ is a postfixpoint of~$\CE_{G'}$. Indeed,
  by~\eqref{item:functionals-monotone}, we have
  $\nu\CE_G=\CE_G(\nu\CE_G)\subseteq\CE_{G'}(\nu\CE_G)$. For~$\CA_G$,
  we show that $G\cap\mu\CA_{G'}$ is a prefixpoint of $\CA_G$. Indeed,
  by~\eqref{item:functionals-monotone}, we have
  $\CA_G(\mu\CA_{G'}\cap G)\subseteq\CA_{G'}(\mu\CA_{G'}\cap
  G)\subseteq\CA_{G'}(\mu\CA_{G'})=\mu\CA_{G'}$, and
  $\CA_G(\mu\CA_{G'}\cap G)\subseteq G$ by the definition of~$\CA_G$.
\end{proof}
\begin{rem}\label{rem:non-duality}
  The reader will note that the functionals $\CA_\Nodes$ and
  $\CE_\Nodes$ fail to be mutually dual, as $\CE_\Nodes$ quantifies
  existentially instead of universally over propositional rules. We
  will show that the well-known commutation of the propositional rules
  implies that the more permissive use of existential quantification
  eventually leads to the same answers (see proof of
  Lemma~\ref{lem:no-propagation}.(\ref{item:inv-AE-Gf})); it allows
  for more economy in the generation of new nodes in the global
  caching algorithm, described next.
\end{rem}
\noindent The global caching algorithm maintains, as global variables,
a set $\Nodes$ of sequents with subsets~$E$ and~$A$ of sequents already
decided as satisfiable or unsatisfiable, respectively. 
\begin{alg}\label{alg:global-caching}
  (Decide $\psi$-satisfiability of $\phi_0$ by global caching.)
  \begin{enumerate}
  \item Initialize $\Nodes=\{\Gamma_0\}$ with $\Gamma_0=\{\phi_0,\psi\}$, and
    $E=A=\emptyset$.
  \item (Expand)\label{step:expand} Select a sequent
    $\Gamma\in \Nodes$ that has children that are not in $\Nodes$, and
    add any number of these children to $\Nodes$. If no sequents with
    missing children are found, go to Step~\ref{step:final-prop}
  \item (Propagate)\label{step:prop} Optionally recalculate $E$ as the greatest fixed
    point $\nu S.\,\CE_\Nodes(S\cup E)$, and $A$ as
    $\mu S.\,\CA_G(S\cup A)$. If $\Gamma_0\in E$, return `yes'; if
    $\Gamma_0\in A$, return `no'. 
  \item Go to Step~\ref{step:expand}.
  \item \label{step:final-prop} Recalculate $E$ as
    $\nu S.\,\CE_G(S\cup E)$; return `yes' if $\Gamma_0\in E$, and
    `no' otherwise.
  \end{enumerate}
\end{alg}
\begin{rem}
  As explained at the beginning of the section, the key feature of the
  global caching algorithm is that it potentially avoids generating
  the full exponential-sized set of tableau sequents by detecting
  satisfiability or unsatisfiability on the fly in the intermediate
  optional propagation steps. The non-determinism in the formulation
  of the algorithm can be resolved arbitrarily, i.e.\ we will see that
  any choice (e.g.\ of which sequents to add in the expansion step and
  whether or not to trigger propagation) leads to correct results;
  thus, it affords room for heuristic optimization. Detecting
  \emph{un}satisfiability in Step~\ref{step:prop} requires previous
  generation of all, in principle exponentially many, children of a
  sequent. This is presumably not necessarily prohibitive in practice,
  as the exponential dependence is only in the number of
  \emph{top-level} modalities in a sequent. As an extreme example, if
  we encode the graded modality $\Diamond_0\phi$ as $\sharp(\phi)>0$
  in Presburger modal logic, then the sequent $\{\Diamond_0^n\top\}$
  ($n$ successive diamonds) induces $2^n$ types but has only two
  children, $\{\Diamond_0^{n-1}\top\}$ and
  $\{\neg\Diamond_0^{n-1}\top\}$.
\end{rem}
\noindent We next prove correctness of the algorithm. As a first step,
we show that a sequent can be added to~$E$ (or to~$A$) in the optional
Step~3 of the algorithm only if it will at any rate end up in~$E$ (or
outside~$E$, respectively) in the final step of the algorithm. To this
end, let~$\Gf$ denote the least set of sequents such that
$\Gamma_0 \in \Gf$ and $\Gf$ contains all children of nodes contained
in $\Gf$, i.e.\ $\children{\Gamma}\subseteq\Gf$ for each
$\Gamma\in\Gf$; that is, at the end of a run of the algorithm without
intermediate propagation steps, we have~$G=\Gf$ and
$E=\nu S.\,\CE_{\Gf}(S)$. We then formulate the claim in the following
invariants:

\begin{lemma}\label{lem:no-propagation}
  At any stage throughout a run of Algorithm~\ref{alg:global-caching} we have
  \begin{enumerate}
  \item\label{item:inv-E-G} $E \subseteq \nu S. \CE_G(S)$
  \item\label{item:inv-A-G} $A \subseteq \mu S. \CA_G(S)$
  \item \label{item:inv-E-Gf} $E \subseteq \nu S . \CE_{\Gf}(S)$
  \item\label{item:inv-A-Gf} $A \subseteq \mu S. \CA_{\Gf}(S)$ 
  \item \label{item:inv-AE-Gf} $A\cap\nu S . \CE_{\Gf}(S)=  \mu S. \CA_{\Gf}(S) \cap\nu S . \CE_{\Gf}(S) = \emptyset$.
\end{enumerate}

\end{lemma}
\noindent In the proof, we use the following simple fixpoint laws (for
which no novelty is claimed):
\begin{lemma}\label{lem:fp-laws}
  Let $X$ be a set, and let $F:\Pow X\to\Pow X$ be monotone w.r.t.\
  set inclusion. Then
  \begin{equation*}
    \nu S.\,F(S\cup \nu S.\,F(S))=\nu S.\,F(S)\quad\text{and}\quad
    \mu S.\,F(S\cup \mu S.\,F(S))=\mu S.\,F(S).
  \end{equation*}
\end{lemma}
\begin{proof}
  In both claims, `$\supseteq$' is trivial; we show `$\subseteq$'.
  For $\nu$, we show (already using `$\supseteq$') that the left-hand side is
  a fixpoint of~$F$:
  \begin{align*}
    & \nu S. F(S \cup \nu S. F(S)) \\
    & =F((\nu S. F(S \cup \nu S. F(S)))\cup (\nu S. F(S)))&& \by{fixpoint unfolding}\\
    & =F(\nu S. F(S \cup \nu S. F(S)))&& \by{$\nu S.\,F(S\cup \nu S.\,F(S))\supseteq\nu S.\,F(S)$}.
  \end{align*}
  For~$\mu$, we show that the right-hand side is a fixpoint of
  $S\mapsto F(S \cup \mu S.F(S))$:
  \begin{equation*}
    F( \mu S.F(S) \cup \mu S.F(S))=F(\mu S.F(S))=\mu S.F(S). \qedhere
  \end{equation*}
\end{proof}
\begin{proof}[Proof (Lemma~\ref{lem:no-propagation})]
  \emph{(\ref{item:inv-E-G}) and~(\ref{item:inv-A-G}):}\/ Clearly, these
  invariants hold \emph{initially}, as~$E$ and~$A$ are initialized
  to~$\emptyset$.

  In \emph{expansion steps}, the invariants are preserved because by
  Lemma~\ref{lem:functionals-monotone}, $\nu S. \CE_G(S)$ and
  $\mu S. \CA_G(S)$ depend monotonically on~$G$.

  Finally, in a \emph{propagation step}, we change $E$ into
  \begin{equation*}
    E' = \nu S. \CE_G(S \cup E) \subseteq \nu S. \CE_G(S \cup \nu S. \CE_G(S)) = \nu S .\CE_G(S),
  \end{equation*}
  where the inclusion is by the invariant for~$E$ and the equality is
  by Lemma~\ref{lem:fp-laws}. Thus, the invariant~(\ref{item:inv-E-G})
  is preserved. Similarly,~$A$ is changed into
  \begin{equation*}
    A' = \mu S. \CA_G(S \cup A) \subseteq \mu S. \CA_G(S \cup \mu S.\CA_G(S)) = \mu S.\CA_G(S)
  \end{equation*}
  where the equality is by Lemma~\ref{lem:fp-laws}, preserving
  invariant~(\ref{item:inv-A-G}).

  \emph{(\ref{item:inv-E-Gf}) and~(\ref{item:inv-A-Gf}):}\/ Immediate
  from (\ref{item:inv-E-G}) and~(\ref{item:inv-A-G}) by
  Lemma~\ref{lem:functionals-monotone}, since $G\subseteq\Gf$ at all
  stages.

  \emph{(\ref{item:inv-AE-Gf}):}\/ Let $\overline{\CA_{\Gf}}$ denote
  the dual of~$\CA_\Gf$, i.e.\
  $\overline{\CA_{\Gf}}(S)=\Gf\setminus\CA_\Gf(\Gf\setminus S)$; that
  is, $\overline{\CA_{\Gf}}$ is defined like~$\CE_\Gf$ except that
  $\overline{\CA_{\Gf}}(S)$ contains a non-state sequent
  $\Gamma\in\Gf\setminus\States$ if \emph{every} propositional rule
  that applies to $\Gamma$ has a conclusion that is contained in~$S$
  (cf.\ Remark~\ref{rem:non-duality}). Then
  $\nu S.\,\overline{\CA_{\Gf}}(S)$ is the complement of
  $\mu S.\,\CA_\Gf(S)$, so by~(\ref{item:inv-A-Gf}) it suffices to
  show $\nu S.\,\CE_\Gf(S)\subseteq\nu
  S.\,\overline{\CA_{\Gf}}(S)$. To this end, we show that
  $\nu S.\,\CE_\Gf(S)$ is a postfixpoint of $\overline{\CA_{\Gf}}$. So
  let $\Gamma\in\nu S.\,\CE_\Gf(S)=\CE_\Gf(\nu
  S.\,\CE_\Gf(S))$. If~$\Gamma$ is a state, then it follows
  immediately that
  $\Gamma\in\overline{\CA_{\Gf}}(\nu S.\,\CE_\Gf(S))$, since the
  definitions of $\CE_\Gf$ and $\overline{\CA_{\Gf}}$ agree on
  containment of states (note that by definition of~$\Gf$,
  $\children{\Gamma}\subseteq\Gf$ for
  every~$\Gamma\in\Gf$). Otherwise, we proceed by induction on the
  size of~$\Gamma$. By definition of~$\CE_\Gf$, there exists a
  conclusion~$\Gamma'\in\nu S.\,\CE_\Gf(S)$ of a propositional
  rule~$R$ applied to~$\Gamma$. By induction,
  $\Gamma'\in\overline{\CA_{\Gf}}(\nu S.\,\CE_\Gf(S))$. Now
  let~$\Delta$ be the set of conclusions of a propositional rule~$R'$
  applied to~$\Gamma$, w.l.o.g.\ distinct from~$R$. Since the
  propositional rules commute, there is a rule application
  to~$\Gamma'$ (corresponding to a postponed application of~$R'$) that
  has a conclusion~$\Gamma''\in\nu S.\,\CE_\Gf(S)$ such
  that~$\Gamma''$ is, via postponed application of~$R$, a conclusion
  of a propositional rule applied to some $\Gamma'''\in\Delta$. Then,
  $\Gamma'''\in\CE_\Gf(\nu S.\,\CE_\Gf(S))=\nu S.\,\CE_\Gf(S)$ by
  definition of~$\CE_\Gf$, showing
  $\Gamma\in\overline{\CA_{\Gf}}(\nu S.\,\CE_\Gf(S))$ as required.
\end{proof}
\noindent Invariants~(\ref{item:inv-E-Gf}) and~(\ref{item:inv-AE-Gf})
in Lemma~\ref{lem:no-propagation} imply that once we prove correctness
for runs of the algorithm that perform propagation only in the last
step~5 (that is, once all children have been added), correctness of
the general algorithm follows. That is, it remains to show that
$\nu S.\,\CE_\Gf(S)$ consists precisely of the satisfiable sequents
in~$\Gf$. We split this claim into two inclusions respectively
corresponding to soundness and completeness in the same way as for the
type elimination algorithm (Section~\ref{sec:type-elim}). The
following statement is analogous to Lemma~\ref{lem:ex-truth}.

\begin{lemma}\label{lem:caching-completeness}
  Let $E$ be a postfixpoint of $\CE_{\Gf}$ and denote by
  $E_s = E \cap \States$ the collection of states contained in
  $E$. Then there is a coalgebra $C=(E_s,\gamma)$ such that
  $E_s \cap \{ \Gamma \mid \Gamma \PLentails \phi \}\subseteq
  \Sem{\phi}_{C}$ for all $\phi \in \Sigma$ (recall that $\PLentails$
  denotes propositional entailment, see
  Definition~\ref{def:prop}). Consequently, whenever $\Gamma \in E$
  and $\phi \in\Gamma$, then~$\phi$ is $\psi$-satisfiable.
\end{lemma}
\begin{proof}
  The proof proceeds similarly to the one of Lemma~\ref{lem:ex-truth}:
  In order to define a suitable $\gamma$, let $\Gamma \in E_s$. By the
  definition of $\CE_{\Gf}$, the one-step pair
  $(\phi_\Gamma,\eta_{E \cap \children{\Gamma}})$ is satisfiable. Let
  $M=(X,\tau,t)$ be a one-step model satisfying
  $(\phi_\Gamma,\eta_{E \cap \children{\Gamma}})$. By the definition of
  $\eta_{E \cap \children{\Gamma}}$, we can then define a function
  $f\colon X \to E \cap \children{\Gamma}$ such that for all $x \in X$ and all
  $\epsilon\hearts \rho \in \Gamma$ we have $\rho\in f(x)$ iff
  $x \in \tau(a_{\epsilon\hearts \rho})$ (noting that by the
  definition of children of~$\Gamma$, $f(x)$ contains either~$\rho$
  or~$\neg\rho$).  Now note that since~$E$ is a postfixpoint of
  $\CE_{\Gf}$, every non-state sequent $\Delta\in E$ has a child
  in~$E$ that is a conclusion of a propositional rule applied
  to~$\Delta$, and hence propositionally entails~$\Land\Delta$. Since
  every propositional rule removes a propositional connective, this
  implies that we eventually reach a state in~$E_s$ from~$\Delta$
  along the child relation; that is, for every $\Delta\in E$ there is
  a state $\Delta'\in E_s$ such that $\Delta'$ propositionally
  entails~$\Land\Delta$.  We can thus prolong $f$ to a function
  $\bar f\colon X \to E_s$ such that
  \begin{equation}\label{eq:barf}
    \bar f (x) \PLentails \rho \quad \mbox{iff} \quad x \in \tau(a_{\epsilon\hearts \rho})
  \end{equation}
  for all $\epsilon\hearts \rho \in \Gamma$ and all $x \in X$.  We now
  define $\gamma(\Gamma) \mathrel{:=} T {\bar f} (t)$, obtaining
  $\gamma\colon E_s \to T E_s$. We will show that
  \begin{equation}
    \Gamma \PLentails \chi \qquad \text{implies} \qquad \Gamma \in
    \Sem{\chi}_C\label{eq:truth}
  \end{equation}
  for all $\chi \in \Sigma$ and all $\Gamma \in E_s$, which implies
  the first claim of the lemma. We proceed by induction on~$\chi$; by
  soundness of propositional reasoning, we immediately reduce to the
  case where~$\chi\in\Gamma$, in which case $\chi$ has the form
  $\chi=\epsilon\hearts\rho$ since~$\Gamma$ is a state. We continue to
  use the data~$M=(X,t,\tau)$, $f$, $\bar f$ featuring in the
  above construction of $\gamma(\Gamma)=T\bar f(t)$. Note again that
  for every $x\in X$, we have by the defining property of children
  of~$\Gamma$ that either $f(x)\PLentails\rho$ or
  $f(x)\PLentails\neg\rho$; since the conclusions of propositional
  rules are propositionally stronger than the premisses, it follows
  that the same holds for $\bar f(x)$. The inductive hypothesis
  therefore implies that $\bar f(x)\in\Sem{\rho}_C$ iff
  $\bar f(x)\PLentails\rho$; combining this with \eqref{eq:barf}, we
  obtain $f^{-1}[\Sem{\rho}_C]=\tau(a_{\epsilon\hearts\rho})$. To
  simplify notation, assume that $\epsilon=1$  (the case
  where~$\epsilon=-1$ being entirely analogous). We then have to
  show $\gamma(\Gamma)\in\Sem{\hearts}_{E_s}(\Sem{\rho}_C)$, which by
  naturality of~$\Sem{\hearts}$ is equivalent to
  $t\in\Sem{\hearts}_X(f^{-1}[\Sem{\rho}_C])=\Sem{\hearts}_X(\tau(a_{\hearts\rho}))$,
  where the equality is by the preceding calculation. But
  $t\in\Sem{\hearts}_X(\tau(a_{\hearts\rho}))$ follows from
  $M\models(\phi_\Gamma,\eta_{E\cap \children{\Gamma}})$ and $\hearts\rho\in\Gamma$ by
  the definition of~$\phi_\Gamma$.

  The second claim of the lemma is now immediate for states
  $\Gamma \in E_s$.  As indicated above, all other sequents
  $\Gamma \in E\setminus E_s$ can be transformed into some
  $\Gamma' \in E_s$ using the propositional rules, in which
  case~$\Gamma'$ propositionally entails all $\rho\in\Gamma$; thus,
  satisfiability of $\Gamma'$ implies satisfiability of
  all~$\rho\in\Gamma$.
\end{proof}
\noindent Lemma~\ref{lem:caching-completeness} ensures completeness of
the algorithm, i.e.\ whenever the algorithm terminates with 'yes',
then $\phi_0$ is $\psi$-satisfiable. For soundness (i.e.\ the
converse implication, the algorithm answers `yes' if~$\phi_0$ is
$\psi$-satisfiable) we proceed similarly as for Lemma
\ref{lem:realization}:

\begin{lemma}\label{lem:caching-soundness} 
The set of $\psi$-satisfiable sequents contained in $\Gf$ is a
post-fixpoint of $\CE_{\Gf}$.
\end{lemma}

\begin{proof}
  Let $S$ be the set of $\psi$-satisfiable sequents in $G_f$. We have
  to show that $S \subseteq \CE_{\Gf}(S)$; so let $\Gamma \in
  S$. If~$\Gamma$ is not a state, then to show $\Gamma\in\CE_{\Gf}(S)$
  we have to check that some propositional rule that applies
  to~$\Gamma$ has a $\psi$-satisfiable conclusion that is moreover
  contained in $\Gf$; this is easily verified by inspection of the
  rules, noting that all children of~$\Gamma$ are in~$\Gf$. Now
  suppose that~$\Gamma$ is a state; we then have to show that the
  one-step pair $(\phi_\Gamma, \eta_{S \cap \children{\Gamma}})$ is
  one-step satisfiable. Let~$x$ be a state in a $\psi$-model
  $C=(X,\gamma)$ such that $x\models_C\Gamma$. We construct a one-step
  model of $(\phi_\Gamma, \eta_{S \cap \children{\Gamma}})$ from~$C$
  in the same way as in the proof of Lemma~\ref{lem:realization}. The
  only point to note additionally is that for every $y\in X$, we have
  some $\Delta\in S\cap\children{\Gamma}$ such that
  $y\models_C\Delta$, namely
  $\Delta=\{\epsilon\rho\mid\epsilon'\hearts\rho\in\Gamma,y\models_C\epsilon\rho\}$
  (where $\epsilon$ and $\epsilon'$ range over $\{-1,1\}$).
\end{proof}
\noindent Summing up, we have
\begin{theorem}\label{thm:global-cache}
  If the strict one-step satisfiability problem of~$\Lambda$ is in
  \ExpTime, then the global caching algorithm decides satisfiability
  under global assumptions in exponential time.
\end{theorem}
\begin{proof}
  Correctness is by Lemma~\ref{lem:caching-soundness} and
  Lemma~\ref{lem:caching-completeness}, taking into account the
  reduction to runs without intermediate propagation according to
  Lemma~\ref{lem:no-propagation}. It remains to analyse run time; this
  point is similar as in Lemma~\ref{lem:type-elim-time}: There are at
  only exponentially many sequents, so there can be at most
  exponentially many expansion steps, and the fixpoint calculations in
  the propagation steps run through at most exponentially many
  iterations. The run time analysis of a single fixpoint iteration
  step is essentially the same as in Lemma~\ref{lem:type-elim-time},
  using that strict one-step satisfiability is in \ExpTime for state
  sequents; for non-state sequents~$\Gamma$ just note that there are
  only polynomially many conclusions of propositional rules arising
  from~$\Gamma$, which need to be compared with at most exponentially
  many existing nodes.
\end{proof}

\section{Concrete Algorithm}\label{sec:concrete-alg}

In the following we provide a more concrete description of the global
caching algorithm, which does not use the computation of least and
greatest fixpoints as primitive operators.  The algorithm closely
follows Liu and Smolka's well-known algorithm for fixpoint computation
in what the authors call ``dependency graphs''~\cite{lism98:simp}; in
our case, these structures are generated by the derivation rules. The
main difference between the algorithm described below and Liu and
Smolka's is caused by the treatment of ``modal'' sequents, i.e.\
states, as the condition that these sequents need to satisfy is not
expressible purely as a reachability property.

As in the previous section we work with a closed set $\Sigma$ 
(generated by the global assumption~$\psi$ and the target formula
$\phi_0$) and \emph{(tableau) sequents}, i.e.\ arbitrary subsets
$\Gamma,\Theta\subseteq\Sigma$, understood conjunctively.  We continue
to write $\Seqs=\Pow\Sigma$ for the set of sequents, and $\States$ for
the set of states, i.e.\ sequents consisting of modal literals only
(recall that we take propositional atoms as nullary operators).

The set $\Seqs$ of sequents carries a hypergraph structure 
$E \subseteq \Seqs \times \Pow (\Seqs)$ that contains
\begin{itemize}
\item for each $\Gamma \in \States$ the pair
  $(\Gamma,\children{\Gamma})$ (recall that  $\children{\Gamma} \subseteq \Seqs$
  denotes the set of children of~$\Gamma$); and
 \item for each
 $\Gamma \in \Seqs \setminus \States$ the set of pairs $\{ (\Gamma,\Delta) \mid \Gamma/\Delta \mbox{ a propositional rule applicable to } \Gamma \}$.
\end{itemize}
In the following we write~$E_M$ for the ``modal'' part of~$E$ induced
by the state-child relationships as per the first bullet point, and
$E_P$ for the part of~$E$ induced by the propositional rules as per
the second bullet point (so $E$ is the disjoint union of~$E_m$ and~$E_p$).

Our algorithm maintains a {\em partial} function
$\alpha: \Seqs \to \{ 0,1\}$ that maps a sequent to $0$ if it is not
$\psi$-satisfiable, to $1$ if it is $\psi$-satisfiable and is
undefined in case its satisfiability cannot be determined yet.  In the
terminology of the previous section $\alpha$ should have the following
properties:
\begin{itemize}
 \item $\alpha(\Gamma) = 1$ iff $\Gamma \in \nu X .\, \CE_G(X)$ and
 \item $\alpha(\Gamma) = 0$ iff $\Gamma \in \mu X.\, \CA_G(X)$
\end{itemize}
where $G$ denotes the set of sequents for which $\alpha$ is defined.
The idea of computing a partial function is that this allows
determining $\psi$-satisfiability of a given sequent without exploring
the full hypergraph. We will now describe an algorithm for computing
$\alpha$ that is inspired by Liu and Smolka's \emph{local}
algorithm~\cite[Figures~3,4]{lism98:simp} and then show its
correctness.

\begin{alg}
\label{alg:concrete} Concrete Global Caching
\begin{algorithmic}
\State Initialize $\alpha$ to be undefined everywhere;
    \State $\alpha(\Gamma_0) \coloneqq 1$; $D(\Gamma_0) = \emptyset$, $W \coloneqq \{ (\Gamma_0,\Delta) \mid   (\Gamma_0,\Delta)  \in E\}$;
    \While{$W \not= \emptyset$}
    \State Pick $e = (\Gamma,\Delta) \in W$;
    \State $W \coloneqq W - \{e\}$;
     \If{$ \exists \Gamma' \in \Delta .\, \text{($\alpha(\Gamma')$ is undefined)}$}  \Comment{Expansion step}
    \State Pick non-empty $U \subseteq \{\Gamma' \in \Delta \mid  \alpha(\Gamma') \mbox{ undefined}\}$;
    \State For each $\Gamma' \in U$ put $\alpha(\Gamma') \coloneqq 1$, $D(\Gamma') \coloneqq \emptyset$, $W = W \cup \{ (\Gamma',\Delta') \mid   (\Gamma',\Delta') \in E \}$;
    \EndIf
    \If{$e \in E_P$}  \Comment{Propagation step}
    \If{$ \forall \Gamma' \in \Delta.\, \alpha(\Gamma') = 0$} \Comment{Case $\Gamma \not\in \States$}
    \State $\alpha(\Gamma) \coloneqq 0$; $W \coloneqq W \cup D(\Gamma)$; $D(\Gamma) \coloneqq \emptyset$;
    \ElsIf{$ \exists \Gamma' \in \Delta.\, \alpha(\Gamma') = 1$}
    \State pick $\Gamma' \in \Delta$ s.t. $\alpha(\Gamma') = 1$ and put $D(\Gamma') \coloneqq D(\Gamma') \cup \{(\Gamma,\Delta)\}$;
    \State $W \coloneqq W - \{(\Gamma'',\Delta'') \in W \mid \Gamma'' == \Gamma \}$;
    \EndIf
    \ElsIf{$e \in E_M$}  \Comment{Propagation step}
    \State $S_0 \coloneqq \{ \Gamma' \in \Delta \mid \alpha(\Gamma') ==0 \}$ ; $S_1 \coloneqq \{ \Gamma' \in \Delta \mid \alpha(\Gamma') == 1\}$ \Comment{Case $\Gamma \in \States$}
    \If{$\Delta == S_0 \cup S_1$ and $(\phi_\Gamma,\eta_{S_1})$ is not one-step satisfiable}
    \State $\alpha(\Gamma) \coloneqq 0$; $W \coloneqq W \cup D(\Gamma)$; $D(\Gamma) \coloneqq \emptyset$;
    \ElsIf{$(\phi_\Gamma,\eta_{S_1})$ is one-step satisfiable}
    \For{$\Gamma' \in S_1$} 
    $D(\Gamma') \coloneqq D(\Gamma') \cup \{(\Gamma,\Delta)\}$;
    \EndFor
   \ElsIf{$\Delta \not= S_0 \cup S_1$}  $W \coloneqq W \cup \{e\}$;
    \EndIf
    \EndIf
    \EndWhile
\end{algorithmic}
\end{alg}
\begin{rem}
  In Algorithm~\ref{alg:concrete}, hyperedges should be understood as
  represented symbolically, i.e.\ either by describing matches of
  propositional rules or by marking a hyperedge as modal (which determines
  the hyperedge uniquely given the source node). This serves in particular
  to avoid having to create all of the exponentially many children of
  a state node at once. Target nodes $\Gamma'\in\Delta$ of hyperedges
  $(\Gamma,\Delta)$ are generated explicitly only once they are picked
  from~$\Delta$ in the expansion step (the propagation step only
  accesses nodes that are already generated).
\end{rem}
\noindent We proceed to show correctness of
Algorithm~\ref{alg:concrete} and establish a precise connection to our
global caching algorithm. First we need a couple of lemmas that
establish key invariants of the algorithm. Note that the current state
of a run of the algorithm can be characterized by the triple
$(\alpha,D,W)$ where $\alpha$ is the current (partial) labelling of
sequents, $D$ assigns to any given sequent $\Gamma$ a set of
hyperedges that need to be investigated if the $\alpha$-value of $\Gamma$
changes, and $W$ contains the set of hyperedges that the algorithm
still has to check. The algorithm terminates when it reaches a state
of the form $(\alpha,D,\emptyset)$, i.e.\ when there are no edges left
to be checked.  Given a state $s = (\alpha,D,W)$ of the algorithm, we
put $G^s_i \coloneqq \{ \Gamma \in \Seqs \mid \alpha(\Gamma) = i \}$
for $i=0,1$, and $G^s = G^s_0 \cup G^s_1$ (so $G^s$ is the domain of
definition of~$\alpha$).

\begin{lemma}\label{lem:unsat}
    Let $\Gamma \in \Seqs$ and suppose $s=(\alpha,D,W)$ is a state reached during execution of the algorithm.  Then 
    $\alpha(\Gamma) = 0$ implies that $\Gamma \in \mu X.\, \CA_{G^s}(X)$ and therefore, by  Lemma~\ref{lem:no-propagation}(5) and Lemma~\ref{lem:caching-soundness},  the sequent  $\Gamma$ is not 
    $\psi$-satisfiable.
\end{lemma}
\begin{proof}
  First note that once $\alpha(\Gamma) = 0$ for some sequent $\Gamma$,
  the value $\alpha(\Gamma)$ will not change any more throughout the
  run of the algorithm, as the only moment when a sequent~$\Gamma$ is
  assigned value~$1$ is when~$\Gamma$ is newly added to the domain of
  $\alpha$. Since $G_s$ can only grow during a run of the algorithm
  and by Lemma~\ref{lem:functionals-monotone},
  $\Gamma \in \mu X.\, \CA_{G^s}(X)$ depends monotonically on $G_s$,
  it suffices to establish the invariant for the point where
  $\alpha(\Gamma)$ is set to~$0$. So suppose that this happens while
  $e = (\Gamma,\Delta)$ is processed, with the state being
  $s=(\alpha,D,W)$ before and $s'=(\alpha',D',W')$ after
  processing~$e$. Suppose that~$s$ satisfies the claimed invariant; we
  have to show that~$s'$ satisfies it as well. We do this for the case
  where $e \in E_P$; the case $e \in E_M$ is completely
  analogous. 
  Since $e \in E_P$, the reason for setting $\alpha'(\Gamma) = 0$ is
  that for all $\Gamma' \in \Delta$ we have $\alpha(\Gamma') = 0$ --
  in other words, we have $\Gamma \in \CA_{G^{s}}(G^{s}_0)$. This
  implies $\Gamma \in \CA_{G^{s'}}(G^{s}_0)$ by
  Lemma~\ref{lem:functionals-monotone} as $G^{s} \subseteq G^{s'}$.
  By assumption on~$s$, we have
  $G^{s}_0 \subseteq \mu X.\, \CA_{G^{s}}(X) \subseteq\mu X.\,
  \CA_{G^{s'}}(X)$, again using
  Lemma~\ref{lem:functionals-monotone} in the second
  step. Monotonicity of $\CA_{G^{s'}}$ now yields
  $$\Gamma \in  \CA_{G^{s'}}(G^{s}_0) \subseteq  \CA_{G^{s'}}( \mu X.\, \CA_{G^{s'}}(X)) = \mu X.\, \CA_{G^{s'}}(X) $$
  as required.
\end{proof}
\noindent 
The following technical lemma follows by inspecting the details of the
algorithm:
\begin{lemma}\label{lem:inv}
    Suppose  $s=(\alpha,D,W)$ is a state reached during execution of the algorithm.
    Then for all $\Gamma \in G^s_1$ and all $(\Gamma,\Delta) \in E$ precisely one of the following holds:
    \begin{itemize}
     \item $(\Gamma,\Delta) \in W$ or 
     \item $\Gamma \not\in \States$ and there is $(\Gamma,\Delta') \in E_P$ with $(\Gamma,\Delta') \in D(\Gamma'')$ for some $\Gamma'' \in  \Delta'$ or
     \item $\Gamma \in \States$  and $(\phi_\Gamma,\eta_{S})$ is one-step satisfiable with
     $S= \{\Gamma' \in \Delta \mid (\Gamma,\Delta) \in D(\Gamma')\}$
     \end{itemize}
     We also note that $D(\Gamma) \not= \emptyset$ implies $\alpha(\Gamma) = 1$.
\end{lemma}
\noindent Correctness of the algorithm is established in the following
theorem.
\begin{theorem}
  When Algorithm~\ref{alg:concrete} terminates at
  $s=(\alpha,D,\emptyset)$ then for all $\Gamma \in \Seqs$ we have:
  \begin{enumerate}
  \item $\alpha(\Gamma) = 0$ implies $\Gamma \in \mu X.\, \CA_{G^s}(X)$ and thus
    $\Gamma$ is not $\psi$-satisfiable.
  \item $\alpha(\Gamma) = 1$ implies $\Gamma \in \nu X.\, \CE_{G^s}(X)$ and thus
    $\Gamma$ is $\psi$-satisfiable.
  \end{enumerate}
\end{theorem}
\begin{proof}
 The first claim is immediate by Lemma~\ref{lem:unsat}. 
 For the second claim it suffices to prove that~$G^s_1$ is included in the greatest fixpoint of $\CE_{G^s}(X)$  - the claim concerning $\psi$-satisfiability of~$\Gamma$ then follows
 from Lemmas~\ref{lem:functionals-monotone} and~\ref{lem:caching-completeness} in the previous section. 
 It suffices to show that $G^s_1$ is a post-fixpoint of $\CE_{G^{s}}$ -- but this follows immediately from
 Lemma~\ref{lem:inv} together with $W = \emptyset$ and $\{ \Gamma \mid D(\Gamma) \not= \emptyset \} \subseteq G^s_1$.
 \end{proof}

 \noindent Algorithm~\ref{alg:concrete} is closely related to
 Algorithm~\ref{alg:global-caching}: Both algorithms explore the
 collection of sequents that are ``reachable'' from $\Gamma_0$, making
 non-deterministic choices concerning which sequents to expand next. A
 crucial difference to Algorithm~\ref{alg:global-caching} is that
 Algorithm~\ref{alg:concrete} contains a concrete description of how
 to compute the fixpoints of $\CE$ and $\CA$ by successively updating
 the labelling function; to this end, it imposes a more definite
 strategy regarding propagation by enforcing a propagation step after
 every expansion step. We conclude by providing an estimate of the
 complexity of the algorithm:

\begin{proposition}
       If the strict one-step satisfiability problem of~$\Lambda$ is in
  \ExpTime, then Algorithm~\ref{alg:concrete} decides  satisfiability
  under global assumptions in exponential time. 
\end{proposition}
\begin{proof}
  To get the upper bound, we observe first that each hyperedge
  $e=(\Gamma, \Delta) \in E_M$ will be checked at most
  $2\cdot|\Delta|$ times by the algorithm: after $e$ has been added to
  $W$ it could be tested up to $|\Delta|$ times (in the worst case,
  until all of the children in~$\Delta$ have been added to the domain
  of $\alpha$) and then again each time the status of one of the
  children in~$\Delta$ changes. 
  Similarly, each hyperedge $e=(\Gamma, \Delta) \in E_P$
  will be checked at most $|\Delta|+1$ times (each time when
  the status of one of the children changes).
  The \ExpTime bound then follows from
  the observation that (i) the hypergraph is exponential in the size of the input, (ii) for $\Gamma \in \States$ there is exactly one
  edge $(\Gamma, \Delta) \in E_M$ and (iii) for each $\Gamma \in \Seqs \setminus \States$ the algorithm only verifies one hyperedge of the 
  form $(\Gamma,\Delta) \in E_P$. 
\end{proof}

\section{Nominals}\label{sec:nominals}

A key feature of \emph{hybrid logic}~\cite{ArecesTenCate07} as an
extension of modal logic are \emph{nominals}, which are special atomic
predicates that are semantically restricted to hold in exactly one
state, and hence uniquely designate a state. Nominals form part of
many relational description logics (recognizable by the
letter~$\mathcal O$ in the standard naming scheme)~\cite{BaaderEA03},
where they serve as expressive means to express facts involving
specific individuals -- for instance, using nominals, concepts over an
ontology of music can not only speak about the notion of composer in
general, but also concretely about Mozart and Stockhausen. We proceed
to discuss how to extend some of the above results to cover
coalgebraic hybrid logic, i.e.\ the extension of coalgebraic modal
logic with nominals in the standard sense. Specifically, we show that
the generic \ExpTime upper bound for reasoning under global
assumptions (Theorem~\ref{thm:exptime}) remains true in presence of
nominals; we leave the design of a global caching algorithm for this
setting as an open problem (for the case where a complete set of modal
tableau rules in the sense recalled in Remark~\ref{rem:rules} is
available, we have presented such an algorithm in previous
work~\cite{GoreEA10b}).

\textbf{Syntactically}, we introduce a set $\Noms$ of \emph{nominals}
$i,j,\dots$, i.e.\ names for individual states, and work with an
extended set $\FLang(\Noms,\Lambda)$ of \emph{hybrid}
formulae~$\phi,\psi$, defined by the grammar
\begin{equation*}
  \FLang(\Noms,\Lambda)\owns\phi,\psi::= \bot\mid \phi\land\psi\mid
  \neg\phi\mid \hearts(\phi_1,\dots,\phi_n)\mid i\mid @_i\phi\qquad (\hearts\in\Lambda\text{ $n$-ary}, i\in\Noms);
\end{equation*}
that is, nominals may be used as atomic formulae and within
\emph{satisfaction operators} $@_i$, with $@_i\phi$ stating that the
state denoted by $i$ satisfies $\phi$. (We explicitly do not include
local binding $\downarrow$, with formulae ${\downarrow}i.\,\phi$ read
`$\phi$ holds if~$i$ is changed to denote the present state', which
would lead to undecidability~\cite{ArecesEA99}.)

\textbf{Semantically}, we work with \emph{hybrid models} $\CM=(C,\pi)$
consisting of a $T$-coalgebra $C=(X,\gamma)$ and an assignment of a
singleton set $\pi(i)\subseteq X$ to each nominal $i\in\Noms$. We
write $\models_\CM$ for the satisfaction relation between states~$x$
in hybrid models~$\CM=(C,\pi)$ and hybrid formulae, defined by
\begin{align*}
  x &\models_\CM i&& \hspace{-4em}\text{iff}\quad x\in\pi(i)\\
  x &\models_\CM @_i\phi && \hspace{-4em}\text{iff}
      \quad y\models_\CM\phi\quad\text{for the unique $y\in\pi(i)$},
\end{align*}
and otherwise the same clauses as $\models_C$
(Section~\ref{sec:colog}). Similarly as for the purely modal logic, we
sometimes refer to these data just as the coalgebraic hybrid
logic~$\Lambda$.
\begin{example}
  We illustrate how the presence of nominals impacts on logical
  consequence.
  \begin{enumerate}
  \item In Presburger modal logic, the formula
  \begin{equation*}
    @_i(\sharp(i)>\sharp(p)),
  \end{equation*}
  with~$i$ a nominal and~$p$ a propositional atom, says that state~$i$
  has higher transition weight to itself than to states
  satisfying~$p$. One consequence of this formula is
  \begin{equation*}
    @_i\neg p.
  \end{equation*}
\item In probabilistic modal logic, the formula
  \begin{equation*}
    @_i(w(j)>w(\neg j)\land w(k)\ge w(\neg k)),
  \end{equation*}
  with nominals $i,j,k$, says that from state~$i$, we reach state~$j$
  with probability strictly greater than $1/2$, and state~$k$ with
  probability at least~$1/2$. From this, we conclude that $j=k$, i.e.\
  \begin{equation*}
    @_jk.
  \end{equation*}
\end{enumerate}
\end{example}

\begin{rem}
  In the presence of nominals, the equivalence of the Kripke semantics
  and multigraph semantics of Presburger modal logic
  (Lemma~\ref{lem:multi-vs-kripke}) breaks down: For a nominal~$i$,
  the formula $\sharp(i)>1$ is satisfiable in multigraph semantics but
  not in Kripke semantics. Using global assumptions, we can however
  encode Kripke semantics into multigraph semantics, by extending the
  global assumption~$\psi$ with additional conjuncts $\sharp(i)\le 1$
  for all nominals $i$ appearing either in $\psi$ or in the target
  formula~$\phi_0$. We therefore continue to use multigraph semantics
  for Presburger hybrid logic.
\end{rem}

\begin{rem}\label{rem:univ-mod-hybrid}
  As in the case of coalgebraic modal logic
  (Remark~\ref{rem:univ-mod}), satisfiability under global assumptions
  in coalgebraic hybrid logic is mutually reducible with plain
  satisfiability in an extended logic featuring the universal
  modality~$\univbox$, with the same syntax and semantics as in
  Remark~\ref{rem:univ-mod}. The non-trivial reduction (from the
  universal modality to global assumptions) works slightly differently
  than in the modal case, due to the fact that we cannot just take
  disjoint unions of hybrid models: Like before, let
  $\univbox\psi_1,\dots,\univbox\psi_n$ be the $\univbox$-subformulae
  of the target formula~$\phi$ (now in coalgebraic hybrid logic with
  the universal modality), and guess a subset
  $U\subseteq\{1,\dots,n\}$, inducing a map $\chi\mapsto\chi[U]$
  eliminating $\univbox$ from subformulae~$\chi$ of~$\phi$ as in
  Remark~\ref{rem:univ-mod}. Then check that $\phi[U]$ is satisfiable
  under the global assumption
  \begin{equation*}
    \psi_U = \Land_{k\in U}\psi_k[U]\land\Land_{k\in\{1,\dots,n\}\setminus U} (i_k\to\neg\psi_k[U])
  \end{equation*}
  where the $i_k$ are fresh nominals. It is easy to see that this
  non-deterministic reduction is correct, i.e.\ that $\phi$ is
  satisfiable iff $\phi[U]$ is $\psi_U$-satisfiable for some~$U$.
\end{rem}
\noindent 
A consequence of Remark~\ref{rem:univ-mod-hybrid} is that for purposes
of estimating the complexity of satisfiability under global
assumptions, we can eliminate satisfaction operators: Using the
universal modality $\univbox$, we can express $@_i\phi$ as
$\univbox(i\to\phi)$. We will thus consider only the language without
satisfaction operators in the following.  For a further reduction, we
say that the global assumption~$\psi$ is \emph{globally satisfiable}
if $\top$ is $\psi$-satisfiable, i.e.\ if there exists a non-empty
$\psi$-model. Then note that $\phi_0$ is $\psi$-satisfiable iff
$\psi\land(i\to\phi_0)$ is globally satisfiable for a fresh
nominal~$i$; so we can forget about the target formula and just
consider global satisfiability.

We proceed to adapt the type elimination algorithm of
Section~\ref{sec:type-elim} to this setting. Fix a global
assumption~$\psi$ to be checked for global satisfiability, and
let~$\Sigma$ be the closure of $\{\psi\}$.
\begin{definition}
  For $i\in\Noms\cap\Sigma$ and $\type\in\types{\psi}$, we say that
  \emph{$i$ has type~$\Gamma$} in a hybrid model $(C,\pi)$ if
  $y\models \type$ for the unique $y\in\pi(i)$.

  A \emph{type assignment} (for~$\Sigma$) is a map
  \begin{equation*}
    \beta\colon \Noms\cap\Sigma\to\types{\psi}.
  \end{equation*}
  We say that~$\beta$ is \emph{consistent} if for all
  $i,j\in\Noms\cap\Sigma$, we have $i\in\beta(j)$ iff
  $\beta(i)=\beta(j)$ (in particular, $i\in\beta(i)$ for
  all~$i\in\Noms\cap\Sigma$). A hybrid model $\CM$
  \emph{satisfies}~$\beta$ if every $i\in\Noms\cap\Sigma$ has
  type~$\beta(i)$ in~$\CM$;~$\beta$ is \emph{$\psi$-satisfiable} if
  there exists a hybrid $\psi$-model that satisfies~$\beta$.
\end{definition}
\noindent (In description logic terminology, we may think of type
assignments as complete ABoxes.) We note the following obvious
properties:
\begin{fact}
  \begin{enumerate}
  \item The formula $\psi$ is globally satisfiable iff there exists a
    $\psi$-satisfiable type assignment for~$\Sigma$.
  \item There are at most exponentially many type
    assignments for~$\Sigma$.
  \item All satisfiable type assignments are consistent.
  \item Consistency of a type assignment can be checked in polynomial
    time.
\end{enumerate}
\end{fact}
\noindent To obtain an upper bound \ExpTime for global satisfiability
of~$\psi$, it thus suffices to show that we can decide in \ExpTime
whether a given consistent type assignment~$\beta$ is
$\psi$-satisfiable. To this end, we form the set
\begin{equation*}
  \types{\beta,\psi}=\beta[\Noms\cap\Sigma]\cup \{\type\in\types{\psi}\mid \type\cap\Noms=\emptyset\}
\end{equation*}
of types -- that is, $\types{\beta,\psi}$ includes the assigned
types~$\beta(i)$ for all nominals~$i\in\Noms\cap\Sigma$, and moreover
all types that do not specify any nominal to be locally satisfied. To
check whether $\beta$ is $\psi$-satisfiable, we then run type
elimination on~$\types{\beta,\psi}$; that is, we compute
$\nu\CE_\beta$ by fixpoint iteration starting from
$\types{\beta,\psi}$, where
\begin{equation*}
 \begin{array}{lcll}
   \CE_\beta\colon &\Pow(\types{\beta,\psi})&\to & \Pow(\types{\beta,\psi})\\[0.3ex]
     & S & \mapsto & \{\type \in S \mid (\phi_\type,\eta_S)\text{ is one-step satisfiable}\}
 \end{array}
\end{equation*}
(in analogy to the functional~$\CE$ according
to~\eqref{eq:elim-functional} as used in the type elimination
algorithm for the purely modal case). We answer `yes' if
$\beta[\Noms\cap\Sigma]\subseteq\nu\CE_\beta$, i.e.\ if no
type~$\beta(i)$ is eliminated, and `no' otherwise.

By the same analysis as in Lemma~\ref{lem:type-elim-time}, we see that
the computation of $\nu\CE_\beta$ runs in exponential time if the
strict one-step satisfiability problem of~$\Lambda$ is in
\ExpTime. Correctness of the algorithm is immediate from the following
fact.
\begin{lemma}
  Let $\beta$ be a consistent type assignment. Then $\beta$ is
  $\psi$-satisfiable iff
  $\beta[\Noms\cap\Sigma]\subseteq\nu\CE_\beta$.
\end{lemma}
\begin{proof}
  Soundness (`only if') follows from
  \begin{equation*}
    R_\beta=\{\type\in \types{\beta,\psi}\mid \type\text{ satisfiable in a hybrid $\psi$-model satisfying~$\beta$}\}
  \end{equation*}
  being a postfixpoint of~$\CE_\beta$; the proof is completely
  analogous to that of Lemma~\ref{lem:realization}.

  To see completeness (`if'), construct a $T$-coalgebra
  $C=(\nu\CE_\beta,\gamma)$ in the same way as in the proof of
  Lemma~\ref{lem:ex-truth}. We turn $C$ into a hybrid model
  $\CM=(C,\pi)$ by putting
  $\pi(i)=\{\type\in\nu\CE_\beta\mid i\in\type\}$, noting that
  $\pi(i)$ is really the singleton $\{\beta(i)\}$ because (i) $\beta$
  is consistent and no type in $\types{\beta,\psi}$ other than the
  $\beta(j)$ (for $j\in\Noms\cap\Sigma$) contains a nominal positively,
  and (ii) $\beta(i)\in\nu\CE_\beta$ by assumption. The truth lemma
  \begin{equation*}
    \Sem{\rho}_C=\hat\rho\cap\nu\CE_\beta=\{\type\in\nu\CE_\beta\mid\rho\in\type\}
  \end{equation*}
  is shown by induction on~$\rho\in\Sigma$. All cases are as in the
  proof of Lemma~\ref{lem:ex-truth}, except for the new case
  $\rho=i\in\Noms$; this case is by construction of~$\pi$.  The truth
  lemma implies that~$\CM$ is a $\psi$-model and satisfies~$\beta$.
\end{proof}
\noindent In summary, we obtain
\begin{theorem}
  If the strict one-step satisfiability problem of~$\Lambda$ is in
  \ExpTime, then satisfiability with global assumptions in the
  coalgebraic hybrid logic~$\Lambda$ is \ExpTime-complete.
\end{theorem}
\begin{rem}
  The \ExpTime algorithm described above is not, of course, one that
  one would wish to use in practice. Specifically, while the
  computation of $\nu\CE_\beta$ for a given consistent type assignment
  can be made practical along the lines of the global caching
  algorithm for the nominal-free case discussed in
  Sections~\ref{sec:caching} and~\ref{sec:concrete-alg}, the initial
  reductions -- elimination of satisfaction operators and, more
  importantly, going through all consistent type assignments -- will
  consistently incur exponential cost. We leave the design of a more
  practical algorithm for coalgebraic hybrid logic with global
  assumptions for future work. In particular, adapting the global
  caching algorithm described in Section~\ref{sec:caching} to this
  setting remains an unsolved challenge: e.g.\ types such as
  $\{i,\phi\}$ and $\{i,\neg \phi\}$, where $i$ is a nominal and
  $\phi$ is any formula such that both $\phi$ and $\neg\phi$ are
  satisfiable, are clearly both satisfiable but cannot both form part
  of a model. The generic algorithm we presented in earlier work with
  Gor\'e~\cite{GoreEA10b} solves this problem by gathering up ABoxes
  along strategies in a tableau game (so that no strategy will win
  that uses both types mentioned above); however, the algorithm
  requires a complete set of tableau-style rules, which is not
  currently available for our two main examples.
\end{rem}
\noindent We record the instantiation of the generic result to our key
examples explicitly:
\begin{example}
  Reasoning with global assumptions in \emph{Presburger hybrid logic}
  and in \emph{probabilistic hybrid logic with polynomial
    inequalities}, i.e.\ in the extensions with nominals of the
  corresponding modal logics as defined in
  Sections~\ref{sec:presburger} and~\ref{sec:prob}, is in \ExpTime.
\end{example}

\section{Conclusions}

We have proved a generic upper bound \ExpTime for reasoning with
global assumptions in coalgebraic modal and hybrid logics, based on a
semantic approach centered around \emph{one-step satisfiability
  checking}. This approach is particularly suitable for logics for
which no tractable sets of modal tableau rules are known; our core
examples of this type are Presburger modal logic and probabilistic
modal logic with polynomial inequalities. The upper complexity bounds
that we obtain for these logics by instantiating our generic results
appear to be new. The upper bound is based on a type elimination
algorithm; additionally, for the purely modal case (i.e.\ in the
absence of nominals), we have designed a global caching algorithm that
offers a perspective for efficient reasoning in practice.

In earlier work on upper bounds \PSpace for plain satisfiability
checking (i.e.~reasoning in the absence of global
assumptions)~\cite{SchroderPattinson08d}, we have used the more
general setting of coalgebraic modal logic over \emph{copointed}
functors. This has allowed covering logics with frame conditions that
are non-iterative~\cite{Lewis74}, i.e.~do not nest modal operators but
possibly have top-level propositional variables, such as the $T$-axiom
$\Box a\to a$ that defines reflexive relational frames; an important
example of this type is Elgesem's logic of agency~\cite{Elgesem97}. We
leave a corresponding generalization of our present results to future
work. A further key point that remains for future research is to
extend the global caching algorithm to cover nominals and satisfaction
operators, combining the methods developed in the present paper with
ideas underlying the existing rule-based global caching algorithm for
coalgebraic hybrid logic~\cite{GoreEA10b}.

\begin{acks}
  We wish to thank Erwin R.\ Catesbeiana for remarks on
  unsatisfiability. Work of the third author supported by the
  \grantsponsor{dfg}{DFG}{www.dfg.de} under the research grant
  \grantnum{dfg}{ProbDL2 (SCHR 1118/6-2)}. 
\end{acks}

\bibliographystyle{ACM-Reference-Format}
\bibliography{coalgml}


\begin{thebibliography}{52}


\ifx \showCODEN    \undefined \def \showCODEN     #1{\unskip}     \fi
\ifx \showDOI      \undefined \def \showDOI       #1{#1}\fi
\ifx \showISBNx    \undefined \def \showISBNx     #1{\unskip}     \fi
\ifx \showISBNxiii \undefined \def \showISBNxiii  #1{\unskip}     \fi
\ifx \showISSN     \undefined \def \showISSN      #1{\unskip}     \fi
\ifx \showLCCN     \undefined \def \showLCCN      #1{\unskip}     \fi
\ifx \shownote     \undefined \def \shownote      #1{#1}          \fi
\ifx \showarticletitle \undefined \def \showarticletitle #1{#1}   \fi
\ifx \showURL      \undefined \def \showURL       {\relax}        \fi
\providecommand\bibfield[2]{#2}
\providecommand\bibinfo[2]{#2}
\providecommand\natexlab[1]{#1}
\providecommand\showeprint[2][]{arXiv:#2}

\bibitem[\protect\citeauthoryear{Areces, Blackburn, and Marx}{Areces
  et~al\mbox{.}}{1999}]%
        {ArecesEA99}
\bibfield{author}{\bibinfo{person}{Carlos Areces}, \bibinfo{person}{Patrick
  Blackburn}, {and} \bibinfo{person}{Maarten Marx}.}
  \bibinfo{year}{1999}\natexlab{}.
\newblock \showarticletitle{A Road-Map on Complexity for Hybrid Logics}. In
  \bibinfo{booktitle}{\emph{Computer Science Logic, {CSL} 1999}}
  \emph{(\bibinfo{series}{LNCS})},
  \bibfield{editor}{\bibinfo{person}{J{\"{o}}rg Flum} {and}
  \bibinfo{person}{Mario Rodr{\'{\i}}guez{-}Artalejo}} (Eds.),
  Vol.~\bibinfo{volume}{1683}. \bibinfo{publisher}{Springer},
  \bibinfo{pages}{307--321}.
\newblock


\bibitem[\protect\citeauthoryear{Areces and ten Cate}{Areces and ten
  Cate}{2007}]%
        {ArecesTenCate07}
\bibfield{author}{\bibinfo{person}{Carlos Areces} {and} \bibinfo{person}{Balder
  ten Cate}.} \bibinfo{year}{2007}\natexlab{}.
\newblock \showarticletitle{Hybrid logics}.
\newblock In \bibinfo{booktitle}{\emph{Handbook of Modal Logic}},
  \bibfield{editor}{\bibinfo{person}{P.~Blackburn}, \bibinfo{person}{J.~van
  Benthem}, {and} \bibinfo{person}{F.~Wolter}} (Eds.).
  \bibinfo{publisher}{Elsevier}, \bibinfo{pages}{821--868}.
\newblock


\bibitem[\protect\citeauthoryear{Awodey}{Awodey}{2010}]%
        {Awodey10}
\bibfield{author}{\bibinfo{person}{Steve Awodey}.}
  \bibinfo{year}{2010}\natexlab{}.
\newblock \bibinfo{booktitle}{\emph{Category Theory} (\bibinfo{edition}{2nd}
  ed.)}.
\newblock \bibinfo{publisher}{Oxford University Press}.
\newblock


\bibitem[\protect\citeauthoryear{Baader, Calvanese, McGuinness, Nardi, and
  Patel-Schneider}{Baader et~al\mbox{.}}{2003}]%
        {BaaderEA03}
\bibfield{editor}{\bibinfo{person}{Franz Baader}, \bibinfo{person}{Diego
  Calvanese}, \bibinfo{person}{Deborah McGuinness}, \bibinfo{person}{Daniele
  Nardi}, {and} \bibinfo{person}{Peter Patel-Schneider}} (Eds.).
  \bibinfo{year}{2003}\natexlab{}.
\newblock \bibinfo{booktitle}{\emph{The Description Logic Handbook}}.
\newblock \bibinfo{publisher}{Cambridge University Press}.
\newblock
\showISBNx{0-521-78176-0}


\bibitem[\protect\citeauthoryear{Baader and Sattler}{Baader and
  Sattler}{1996}]%
        {BaaderSattler96}
\bibfield{author}{\bibinfo{person}{Franz Baader} {and} \bibinfo{person}{Ulrike
  Sattler}.} \bibinfo{year}{1996}\natexlab{}.
\newblock \showarticletitle{Description Logics with Symbolic Number
  Restrictions}. In \bibinfo{booktitle}{\emph{European Conf.\ Artificial
  Intelligence, ECAI 1996}}, \bibfield{editor}{\bibinfo{person}{Wolfgang
  Wahlster}} (Ed.). \bibinfo{publisher}{Wiley}, \bibinfo{pages}{283--287}.
\newblock


\bibitem[\protect\citeauthoryear{B{\'{a}}rcenas and Lavalle}{B{\'{a}}rcenas and
  Lavalle}{2013}]%
        {BarcenasLavalle13}
\bibfield{author}{\bibinfo{person}{Everardo B{\'{a}}rcenas} {and}
  \bibinfo{person}{Jes{\'{u}}s Lavalle}.} \bibinfo{year}{2013}\natexlab{}.
\newblock \showarticletitle{Expressive Reasoning on Tree Structures: Recursion,
  Inverse Programs, {P}resburger Constraints and Nominals}. In
  \bibinfo{booktitle}{\emph{Advances in Artificial Intelligence and its
  Applications, {MICAI} 2013}} \emph{(\bibinfo{series}{LNCS})},
  \bibfield{editor}{\bibinfo{person}{F{\'{e}}lix~Castro Espinoza},
  \bibinfo{person}{Alexander~F. Gelbukh}, {and} \bibinfo{person}{Miguel
  Gonz{\'{a}}lez}} (Eds.), Vol.~\bibinfo{volume}{8265}.
  \bibinfo{publisher}{Springer}, \bibinfo{pages}{80--91}.
\newblock
\showISBNx{978-3-642-45113-3}


\bibitem[\protect\citeauthoryear{Blackburn, de~Rijke, and Venema}{Blackburn
  et~al\mbox{.}}{2001}]%
        {BlackburnEA01}
\bibfield{author}{\bibinfo{person}{Patrick Blackburn}, \bibinfo{person}{Maarten
  de Rijke}, {and} \bibinfo{person}{Yde Venema}.}
  \bibinfo{year}{2001}\natexlab{}.
\newblock \bibinfo{booktitle}{\emph{Modal Logic}}.
\newblock \bibinfo{publisher}{Cambridge University Press}.
\newblock


\bibitem[\protect\citeauthoryear{Book, Long, and Selman}{Book
  et~al\mbox{.}}{1984}]%
        {book-long-selman:npmv}
\bibfield{author}{\bibinfo{person}{Ronald Book}, \bibinfo{person}{Timothy
  Long}, {and} \bibinfo{person}{Alan Selman}.} \bibinfo{year}{1984}\natexlab{}.
\newblock \showarticletitle{Quantitative Relativizations of Complexity
  Classes}.
\newblock \bibinfo{journal}{\emph{{SIAM} J.\ Comput.}} \bibinfo{volume}{13},
  \bibinfo{number}{3} (\bibinfo{year}{1984}), \bibinfo{pages}{461--487}.
\newblock


\bibitem[\protect\citeauthoryear{Canny}{Canny}{1988}]%
        {Canny88}
\bibfield{author}{\bibinfo{person}{John Canny}.}
  \bibinfo{year}{1988}\natexlab{}.
\newblock \showarticletitle{Some Algebraic and Geometric Computations in
  {PSPACE}}. In \bibinfo{booktitle}{\emph{Symposium on Theory of Computing,
  STOC 1988}}. \bibinfo{publisher}{ACM}, \bibinfo{pages}{460--467}.
\newblock


\bibitem[\protect\citeauthoryear{C{\^{\i}}rstea, Kurz, Pattinson,
  Schr{\"{o}}der, and Venema}{C{\^{\i}}rstea et~al\mbox{.}}{2011}]%
        {CirsteaEA11}
\bibfield{author}{\bibinfo{person}{Corina C{\^{\i}}rstea},
  \bibinfo{person}{Alexander Kurz}, \bibinfo{person}{Dirk Pattinson},
  \bibinfo{person}{Lutz Schr{\"{o}}der}, {and} \bibinfo{person}{Yde Venema}.}
  \bibinfo{year}{2011}\natexlab{}.
\newblock \showarticletitle{Modal Logics are Coalgebraic}.
\newblock \bibinfo{journal}{\emph{Comput.\ J.}} \bibinfo{volume}{54},
  \bibinfo{number}{1} (\bibinfo{year}{2011}), \bibinfo{pages}{31--41}.
\newblock


\bibitem[\protect\citeauthoryear{D'Agostino and Visser}{D'Agostino and
  Visser}{2002}]%
        {DAgostinoVisser02}
\bibfield{author}{\bibinfo{person}{Giovanna D'Agostino} {and}
  \bibinfo{person}{Albert Visser}.} \bibinfo{year}{2002}\natexlab{}.
\newblock \showarticletitle{Finality regained: {A} coalgebraic study of
  Scott-sets and multisets}.
\newblock \bibinfo{journal}{\emph{Arch.\ Math.\ Log.}} \bibinfo{volume}{41},
  \bibinfo{number}{3} (\bibinfo{year}{2002}), \bibinfo{pages}{267--298}.
\newblock


\bibitem[\protect\citeauthoryear{Demri and Lugiez}{Demri and Lugiez}{2006}]%
        {DemriLugiez06}
\bibfield{author}{\bibinfo{person}{St{\'{e}}phane Demri} {and}
  \bibinfo{person}{Denis Lugiez}.} \bibinfo{year}{2006}\natexlab{}.
\newblock \showarticletitle{Presburger Modal Logic Is PSPACE-Complete}. In
  \bibinfo{booktitle}{\emph{Automated Reasoning, {IJCAR} 2006}}
  \emph{(\bibinfo{series}{LNCS})}, \bibfield{editor}{\bibinfo{person}{Ulrich
  Furbach} {and} \bibinfo{person}{Natarajan Shankar}} (Eds.),
  Vol.~\bibinfo{volume}{4130}. \bibinfo{publisher}{Springer},
  \bibinfo{pages}{541--556}.
\newblock
\showISBNx{3-540-37187-7}


\bibitem[\protect\citeauthoryear{Demri and Lugiez}{Demri and Lugiez}{2010}]%
        {DemriLugiez10}
\bibfield{author}{\bibinfo{person}{St{\'e}phane Demri} {and}
  \bibinfo{person}{Denis Lugiez}.} \bibinfo{year}{2010}\natexlab{}.
\newblock \showarticletitle{Complexity of Modal Logics with {P}resburger
  Constraints}.
\newblock \bibinfo{journal}{\emph{J.\ Applied Logic}}  \bibinfo{volume}{8}
  (\bibinfo{year}{2010}), \bibinfo{pages}{233--252}.
\newblock


\bibitem[\protect\citeauthoryear{Eisenbrand and Shmonin}{Eisenbrand and
  Shmonin}{2006}]%
        {EisenbrandShmonin06}
\bibfield{author}{\bibinfo{person}{Friedrich Eisenbrand} {and}
  \bibinfo{person}{Gennady Shmonin}.} \bibinfo{year}{2006}\natexlab{}.
\newblock \showarticletitle{Carath{\'e}odory bounds for integer cones}.
\newblock \bibinfo{journal}{\emph{Oper.\ Res.\ Lett.}} \bibinfo{volume}{34},
  \bibinfo{number}{5} (\bibinfo{year}{2006}), \bibinfo{pages}{564--568}.
\newblock


\bibitem[\protect\citeauthoryear{Elgesem}{Elgesem}{1997}]%
        {Elgesem97}
\bibfield{author}{\bibinfo{person}{Dag Elgesem}.}
  \bibinfo{year}{1997}\natexlab{}.
\newblock \showarticletitle{The modal logic of agency}.
\newblock \bibinfo{journal}{\emph{Nordic J.\ Philos.\ Logic}}
  \bibinfo{volume}{2} (\bibinfo{year}{1997}), \bibinfo{pages}{1--46}.
\newblock


\bibitem[\protect\citeauthoryear{Fagin and Halpern}{Fagin and Halpern}{1994}]%
        {FaginHalpern94}
\bibfield{author}{\bibinfo{person}{Ronald Fagin} {and} \bibinfo{person}{Joseph
  Halpern}.} \bibinfo{year}{1994}\natexlab{}.
\newblock \showarticletitle{Reasoning about knowledge and probability}.
\newblock \bibinfo{journal}{\emph{J.\ ACM}} \bibinfo{volume}{41},
  \bibinfo{number}{2} (\bibinfo{year}{1994}), \bibinfo{pages}{340--367}.
\newblock


\bibitem[\protect\citeauthoryear{Fagin, Halpern, and Megiddo}{Fagin
  et~al\mbox{.}}{1990}]%
        {FaginHalpernMegiddo90}
\bibfield{author}{\bibinfo{person}{Ronald Fagin}, \bibinfo{person}{Joseph
  Halpern}, {and} \bibinfo{person}{Nimrod Megiddo}.}
  \bibinfo{year}{1990}\natexlab{}.
\newblock \showarticletitle{A logic for reasoning about probabilities}.
\newblock \bibinfo{journal}{\emph{Inform.\ Comput.}}  \bibinfo{volume}{87}
  (\bibinfo{year}{1990}), \bibinfo{pages}{78--128}.
\newblock


\bibitem[\protect\citeauthoryear{Fine}{Fine}{1972}]%
        {Fine72}
\bibfield{author}{\bibinfo{person}{Kit Fine}.} \bibinfo{year}{1972}\natexlab{}.
\newblock \showarticletitle{In so many possible worlds}.
\newblock \bibinfo{journal}{\emph{Notre Dame J.\ Form.\ Log.}}
  \bibinfo{volume}{13} (\bibinfo{year}{1972}), \bibinfo{pages}{516--520}.
\newblock


\bibitem[\protect\citeauthoryear{Fischer and Ladner}{Fischer and
  Ladner}{1979}]%
        {FischerLadner79}
\bibfield{author}{\bibinfo{person}{Michael Fischer} {and}
  \bibinfo{person}{Richard Ladner}.} \bibinfo{year}{1979}\natexlab{}.
\newblock \showarticletitle{Propositional Dynamic Logic of Regular Programs}.
\newblock \bibinfo{journal}{\emph{J.\ Comput.\ Syst.\ Sci.}}
  \bibinfo{volume}{18}, \bibinfo{number}{2} (\bibinfo{year}{1979}),
  \bibinfo{pages}{194--211}.
\newblock


\bibitem[\protect\citeauthoryear{Fischer and Rosenberg}{Fischer and
  Rosenberg}{1968}]%
        {FischerRosenberg68}
\bibfield{author}{\bibinfo{person}{Michael Fischer} {and}
  \bibinfo{person}{Arnold Rosenberg}.} \bibinfo{year}{1968}\natexlab{}.
\newblock \showarticletitle{Limited Random Access Turing Machines}. In
  \bibinfo{booktitle}{\emph{Switching and Automata Theory, SWAT (FOCS) 1968}}.
  \bibinfo{publisher}{{IEEE} Computer Society}, \bibinfo{pages}{356--367}.
\newblock


\bibitem[\protect\citeauthoryear{Goranko and Passy}{Goranko and Passy}{1992}]%
        {GorankoPassy92}
\bibfield{author}{\bibinfo{person}{Valentin Goranko} {and}
  \bibinfo{person}{Solomon Passy}.} \bibinfo{year}{1992}\natexlab{}.
\newblock \showarticletitle{Using the Universal Modality: Gains and Questions}.
\newblock \bibinfo{journal}{\emph{J.\ Log.\ Comput.}}  \bibinfo{volume}{2}
  (\bibinfo{year}{1992}), \bibinfo{pages}{5--30}.
\newblock


\bibitem[\protect\citeauthoryear{Gor{\'e}, Kupke, and Pattinson}{Gor{\'e}
  et~al\mbox{.}}{2010a}]%
        {GoreEA10a}
\bibfield{author}{\bibinfo{person}{Rajeev Gor{\'e}}, \bibinfo{person}{Clemens
  Kupke}, {and} \bibinfo{person}{Dirk Pattinson}.}
  \bibinfo{year}{2010}\natexlab{a}.
\newblock \showarticletitle{Optimal Tableau Algorithms for Coalgebraic Logics}.
  In \bibinfo{booktitle}{\emph{Tools and Algorithms for the Construction and
  Analysis of Systems, TACAS 2010}} \emph{(\bibinfo{series}{LNCS})},
  Vol.~\bibinfo{volume}{6015}. \bibinfo{publisher}{Springer},
  \bibinfo{pages}{114--128}.
\newblock


\bibitem[\protect\citeauthoryear{Gor{\'e}, Kupke, Pattinson, and
  Schr{\"o}der}{Gor{\'e} et~al\mbox{.}}{2010b}]%
        {GoreEA10b}
\bibfield{author}{\bibinfo{person}{Rajeev Gor{\'e}}, \bibinfo{person}{Clemens
  Kupke}, \bibinfo{person}{Dirk Pattinson}, {and} \bibinfo{person}{Lutz
  Schr{\"o}der}.} \bibinfo{year}{2010}\natexlab{b}.
\newblock \showarticletitle{Global Caching for Coalgebraic Description Logics}.
  In \bibinfo{booktitle}{\emph{Automated Reasoning, IJCAR 2010}}
  \emph{(\bibinfo{series}{LNCS})},
  \bibfield{editor}{\bibinfo{person}{J{\"u}rgen Giesl} {and}
  \bibinfo{person}{Reiner H{\"a}hnle}} (Eds.), Vol.~\bibinfo{volume}{6173}.
  \bibinfo{publisher}{Springer}, \bibinfo{pages}{46--60}.
\newblock
\showISBNx{978-3-642-14202-4}


\bibitem[\protect\citeauthoryear{Gor{\'{e}} and Nguyen}{Gor{\'{e}} and
  Nguyen}{2013}]%
        {GoreNguyen13}
\bibfield{author}{\bibinfo{person}{Rajeev Gor{\'{e}}} {and}
  \bibinfo{person}{Linh~Anh Nguyen}.} \bibinfo{year}{2013}\natexlab{}.
\newblock \showarticletitle{ExpTime Tableaux for {$\mathcal{ALC}$} Using Sound
  Global Caching}.
\newblock \bibinfo{journal}{\emph{J.\ Autom.\ Reasoning}} \bibinfo{volume}{50},
  \bibinfo{number}{4} (\bibinfo{year}{2013}), \bibinfo{pages}{355--381}.
\newblock


\bibitem[\protect\citeauthoryear{Gor{\'{e}} and Postniece}{Gor{\'{e}} and
  Postniece}{2008}]%
        {Gore:2008:EEG}
\bibfield{author}{\bibinfo{person}{Rajeev Gor{\'{e}}} {and}
  \bibinfo{person}{Linda Postniece}.} \bibinfo{year}{2008}\natexlab{}.
\newblock \showarticletitle{An Experimental Evaluation of Global Caching for
  {$\mathcal{ALC}$} (System Description)}. In
  \bibinfo{booktitle}{\emph{Automated Reasoning, {IJCAR} 2008}}
  \emph{(\bibinfo{series}{LNCS})},
  \bibfield{editor}{\bibinfo{person}{Alessandro Armando},
  \bibinfo{person}{Peter Baumgartner}, {and} \bibinfo{person}{Gilles Dowek}}
  (Eds.), Vol.~\bibinfo{volume}{5195}. \bibinfo{publisher}{Springer},
  \bibinfo{pages}{299--305}.
\newblock


\bibitem[\protect\citeauthoryear{Guti{\'{e}}rrez{-}Basulto, Jung, Lutz, and
  Schr{\"{o}}der}{Guti{\'{e}}rrez{-}Basulto et~al\mbox{.}}{2017}]%
        {GutierrezBasultoEA17}
\bibfield{author}{\bibinfo{person}{V{\'{\i}}ctor Guti{\'{e}}rrez{-}Basulto},
  \bibinfo{person}{Jean~Christoph Jung}, \bibinfo{person}{Carsten Lutz}, {and}
  \bibinfo{person}{Lutz Schr{\"{o}}der}.} \bibinfo{year}{2017}\natexlab{}.
\newblock \showarticletitle{Probabilistic Description Logics for Subjective
  Uncertainty}.
\newblock \bibinfo{journal}{\emph{J.\ Artif.\ Intell.\ Res.}}
  \bibinfo{volume}{58} (\bibinfo{year}{2017}), \bibinfo{pages}{1--66}.
\newblock
\urldef\tempurl%
\url{https://doi.org/10.1613/jair.5222}
\showDOI{\tempurl}


\bibitem[\protect\citeauthoryear{Hausmann and Schr{\"{o}}der}{Hausmann and
  Schr{\"{o}}der}{2019}]%
        {HausmannSchroder19}
\bibfield{author}{\bibinfo{person}{Daniel Hausmann} {and} \bibinfo{person}{Lutz
  Schr{\"{o}}der}.} \bibinfo{year}{2019}\natexlab{}.
\newblock \showarticletitle{Optimal Satisfiability Checking for Arithmetic
  {$\mu$}--Calculi}. In \bibinfo{booktitle}{\emph{Foundations of Software
  Science and Computation Structures, {FOSSACS} 2019}}
  \emph{(\bibinfo{series}{LNCS})}, \bibfield{editor}{\bibinfo{person}{Mikolaj
  Bojanczyk} {and} \bibinfo{person}{Alex Simpson}} (Eds.),
  Vol.~\bibinfo{volume}{11425}. \bibinfo{publisher}{Springer},
  \bibinfo{pages}{277--294}.
\newblock


\bibitem[\protect\citeauthoryear{Heifetz and Mongin}{Heifetz and
  Mongin}{2001}]%
        {HeifetzMongin01}
\bibfield{author}{\bibinfo{person}{Aviad Heifetz} {and}
  \bibinfo{person}{Philippe Mongin}.} \bibinfo{year}{2001}\natexlab{}.
\newblock \showarticletitle{Probabilistic logic for type spaces}.
\newblock \bibinfo{journal}{\emph{Games Econ.\ Behav.}}  \bibinfo{volume}{35}
  (\bibinfo{year}{2001}), \bibinfo{pages}{31--53}.
\newblock


\bibitem[\protect\citeauthoryear{Kupke and Pattinson}{Kupke and
  Pattinson}{2010}]%
        {Kupke:2010:MLL}
\bibfield{author}{\bibinfo{person}{Clemens Kupke} {and} \bibinfo{person}{Dirk
  Pattinson}.} \bibinfo{year}{2010}\natexlab{}.
\newblock \showarticletitle{On Modal Logics of Linear Inequalities}. In
  \bibinfo{booktitle}{\emph{Advances in Modal Logic, AiML 2010}},
  \bibfield{editor}{\bibinfo{person}{Lev Beklemishev},
  \bibinfo{person}{Valentin Goranko}, {and} \bibinfo{person}{Valentin
  Shehtman}} (Eds.). \bibinfo{publisher}{College Publications},
  \bibinfo{pages}{235--255}.
\newblock
\showISBNx{978-1-84890-013-4}


\bibitem[\protect\citeauthoryear{Kupke, Pattinson, and Schr{\"{o}}der}{Kupke
  et~al\mbox{.}}{2015}]%
        {KupkeEA15}
\bibfield{author}{\bibinfo{person}{Clemens Kupke}, \bibinfo{person}{Dirk
  Pattinson}, {and} \bibinfo{person}{Lutz Schr{\"{o}}der}.}
  \bibinfo{year}{2015}\natexlab{}.
\newblock \showarticletitle{Reasoning with Global Assumptions in Arithmetic
  Modal Logics}. In \bibinfo{booktitle}{\emph{Fundamentals of Computation
  Theory, {FCT} 2015}} \emph{(\bibinfo{series}{LNCS})},
  \bibfield{editor}{\bibinfo{person}{Adrian Kosowski} {and}
  \bibinfo{person}{Igor Walukiewicz}} (Eds.), Vol.~\bibinfo{volume}{9210}.
  \bibinfo{publisher}{Springer}, \bibinfo{pages}{367--380}.
\newblock


\bibitem[\protect\citeauthoryear{Ladner}{Ladner}{1977}]%
        {Ladner77}
\bibfield{author}{\bibinfo{person}{Richard Ladner}.}
  \bibinfo{year}{1977}\natexlab{}.
\newblock \showarticletitle{The Computational Complexity of Provability in
  Systems of Modal Propositional Logic}.
\newblock \bibinfo{journal}{\emph{{SIAM} J.\ Comput.}} \bibinfo{volume}{6},
  \bibinfo{number}{3} (\bibinfo{year}{1977}), \bibinfo{pages}{467--480}.
\newblock
\urldef\tempurl%
\url{https://doi.org/10.1137/0206033}
\showDOI{\tempurl}


\bibitem[\protect\citeauthoryear{Larsen and Skou}{Larsen and Skou}{1991}]%
        {LarsenSkou91}
\bibfield{author}{\bibinfo{person}{Kim Larsen} {and} \bibinfo{person}{Arne
  Skou}.} \bibinfo{year}{1991}\natexlab{}.
\newblock \showarticletitle{Bisimulation through probabilistic testing}.
\newblock \bibinfo{journal}{\emph{Inf.\ Comput.}} \bibinfo{volume}{94},
  \bibinfo{number}{1} (\bibinfo{year}{1991}), \bibinfo{pages}{1--28}.
\newblock
\showISSN{0890-5401}


\bibitem[\protect\citeauthoryear{Lewis}{Lewis}{1974}]%
        {Lewis74}
\bibfield{author}{\bibinfo{person}{David Lewis}.}
  \bibinfo{year}{1974}\natexlab{}.
\newblock \showarticletitle{Intensional logics without iterative axioms}.
\newblock \bibinfo{journal}{\emph{J.\ Philos.\ Log.}} \bibinfo{volume}{3},
  \bibinfo{number}{4} (\bibinfo{year}{1974}), \bibinfo{pages}{457--466}.
\newblock


\bibitem[\protect\citeauthoryear{Liu and Smolka}{Liu and Smolka}{1998}]%
        {lism98:simp}
\bibfield{author}{\bibinfo{person}{Xinxin Liu} {and} \bibinfo{person}{Scott
  Smolka}.} \bibinfo{year}{1998}\natexlab{}.
\newblock \showarticletitle{Simple linear-time algorithms for minimal fixed
  points}. In \bibinfo{booktitle}{\emph{Automata, Languages and Programming,
  ICALP 1998}} \emph{(\bibinfo{series}{LNCS})},
  \bibfield{editor}{\bibinfo{person}{Kim Larsen}, \bibinfo{person}{Sven Skyum},
  {and} \bibinfo{person}{Glynn Winskel}} (Eds.), Vol.~\bibinfo{volume}{1443}.
  \bibinfo{publisher}{Springer}, \bibinfo{pages}{53--66}.
\newblock


\bibitem[\protect\citeauthoryear{Mio}{Mio}{2011}]%
        {Mio11}
\bibfield{author}{\bibinfo{person}{Matteo Mio}.}
  \bibinfo{year}{2011}\natexlab{}.
\newblock \showarticletitle{Probabilistic Modal {$\mu$}-Calculus with
  Independent Product}. In \bibinfo{booktitle}{\emph{Foundations of Software
  Science and Computational Structures, {FOSSACS} 2011}}
  \emph{(\bibinfo{series}{LNCS})}, \bibfield{editor}{\bibinfo{person}{Martin
  Hofmann}} (Ed.), Vol.~\bibinfo{volume}{6604}. \bibinfo{publisher}{Springer},
  \bibinfo{pages}{290--304}.
\newblock


\bibitem[\protect\citeauthoryear{Myers, Pattinson, and Schr{\"o}der}{Myers
  et~al\mbox{.}}{2009}]%
        {MyersEA09}
\bibfield{author}{\bibinfo{person}{Rob Myers}, \bibinfo{person}{Dirk
  Pattinson}, {and} \bibinfo{person}{Lutz Schr{\"o}der}.}
  \bibinfo{year}{2009}\natexlab{}.
\newblock \showarticletitle{Coalgebraic Hybrid Logic}. In
  \bibinfo{booktitle}{\emph{Foundations of Software Science and Computation
  Structures, FoSSaCS 2009}} \emph{(\bibinfo{series}{LNCS})},
  \bibfield{editor}{\bibinfo{person}{Luca de~Alfaro}} (Ed.),
  Vol.~\bibinfo{volume}{5504}. \bibinfo{publisher}{Springer},
  \bibinfo{pages}{137--151}.
\newblock


\bibitem[\protect\citeauthoryear{Pacuit and Salame}{Pacuit and Salame}{2004}]%
        {PacuitSalame04}
\bibfield{author}{\bibinfo{person}{Eric Pacuit} {and} \bibinfo{person}{Samer
  Salame}.} \bibinfo{year}{2004}\natexlab{}.
\newblock \showarticletitle{Majority Logic}. In
  \bibinfo{booktitle}{\emph{Principles of Knowledge Representation and
  Reasoning, KR 2004}}, \bibfield{editor}{\bibinfo{person}{Didier Dubois},
  \bibinfo{person}{Christopher~A. Welty}, {and} \bibinfo{person}{Mary-Anne
  Williams}} (Eds.). \bibinfo{publisher}{AAAI Press},
  \bibinfo{pages}{598--605}.
\newblock


\bibitem[\protect\citeauthoryear{Papadimitriou}{Papadimitriou}{1981}]%
        {Papadimitriou81}
\bibfield{author}{\bibinfo{person}{Christos Papadimitriou}.}
  \bibinfo{year}{1981}\natexlab{}.
\newblock \showarticletitle{On the complexity of integer programming}.
\newblock \bibinfo{journal}{\emph{J.\ ACM}}  \bibinfo{volume}{28}
  (\bibinfo{year}{1981}), \bibinfo{pages}{765--768}.
\newblock


\bibitem[\protect\citeauthoryear{Pattinson}{Pattinson}{2004}]%
        {Pattinson04}
\bibfield{author}{\bibinfo{person}{Dirk Pattinson}.}
  \bibinfo{year}{2004}\natexlab{}.
\newblock \showarticletitle{Expressive Logics for Coalgebras via Terminal
  Sequence Induction}.
\newblock \bibinfo{journal}{\emph{Notre Dame J.\ Formal Logic}}
  \bibinfo{volume}{45} (\bibinfo{year}{2004}), \bibinfo{pages}{19--33}.
\newblock


\bibitem[\protect\citeauthoryear{Pratt}{Pratt}{1979}]%
        {Pratt79}
\bibfield{author}{\bibinfo{person}{Vaughan Pratt}.}
  \bibinfo{year}{1979}\natexlab{}.
\newblock \showarticletitle{Models of Program Logics}. In
  \bibinfo{booktitle}{\emph{Foundations of Computer Science, FOCS 1979}}.
  \bibinfo{publisher}{{IEEE} Comp.\ Soc.}, \bibinfo{pages}{115--122}.
\newblock


\bibitem[\protect\citeauthoryear{Rutten}{Rutten}{2000}]%
        {Rutten00}
\bibfield{author}{\bibinfo{person}{Jan Rutten}.}
  \bibinfo{year}{2000}\natexlab{}.
\newblock \showarticletitle{Universal Coalgebra: A Theory of Systems}.
\newblock \bibinfo{journal}{\emph{Theor.\ Comput.\ Sci.}}
  \bibinfo{volume}{249} (\bibinfo{year}{2000}), \bibinfo{pages}{3--80}.
\newblock


\bibitem[\protect\citeauthoryear{Schrijver}{Schrijver}{1986}]%
        {Schrijver86}
\bibfield{author}{\bibinfo{person}{Alexander Schrijver}.}
  \bibinfo{year}{1986}\natexlab{}.
\newblock \bibinfo{booktitle}{\emph{Theory of linear and integer programming}}.
\newblock \bibinfo{publisher}{Wiley Interscience}.
\newblock


\bibitem[\protect\citeauthoryear{Schr{\"o}der}{Schr{\"o}der}{2007}]%
        {Schroder07}
\bibfield{author}{\bibinfo{person}{Lutz Schr{\"o}der}.}
  \bibinfo{year}{2007}\natexlab{}.
\newblock \showarticletitle{A Finite Model Construction for Coalgebraic Modal
  Logic}.
\newblock \bibinfo{journal}{\emph{J.\ Log.\ Algebr.\ Prog.}}
  \bibinfo{volume}{73} (\bibinfo{year}{2007}), \bibinfo{pages}{97--110}.
\newblock


\bibitem[\protect\citeauthoryear{Schr{\"{o}}der}{Schr{\"{o}}der}{2008}]%
        {Schroder08}
\bibfield{author}{\bibinfo{person}{Lutz Schr{\"{o}}der}.}
  \bibinfo{year}{2008}\natexlab{}.
\newblock \showarticletitle{Expressivity of coalgebraic modal logic: The limits
  and beyond}.
\newblock \bibinfo{journal}{\emph{Theor.\ Comput.\ Sci.}}
  \bibinfo{volume}{390}, \bibinfo{number}{2-3} (\bibinfo{year}{2008}),
  \bibinfo{pages}{230--247}.
\newblock


\bibitem[\protect\citeauthoryear{Schr{\"{o}}der and Pattinson}{Schr{\"{o}}der
  and Pattinson}{2006}]%
        {SchroderPattinson06}
\bibfield{author}{\bibinfo{person}{Lutz Schr{\"{o}}der} {and}
  \bibinfo{person}{Dirk Pattinson}.} \bibinfo{year}{2006}\natexlab{}.
\newblock \showarticletitle{{PSPACE} Bounds for Rank-1 Modal Logics}. In
  \bibinfo{booktitle}{\emph{Logic in Computer Science, LICS 2006}}.
  \bibinfo{publisher}{{IEEE} Comp.\ Soc.}, \bibinfo{pages}{231--242}.
\newblock


\bibitem[\protect\citeauthoryear{Schr{\"o}der and Pattinson}{Schr{\"o}der and
  Pattinson}{2008}]%
        {SchroderPattinson08d}
\bibfield{author}{\bibinfo{person}{Lutz Schr{\"o}der} {and}
  \bibinfo{person}{Dirk Pattinson}.} \bibinfo{year}{2008}\natexlab{}.
\newblock \showarticletitle{Shallow models for non-iterative modal logics}. In
  \bibinfo{booktitle}{\emph{Advances in Artificial Intelligence, KI 2008}}
  \emph{(\bibinfo{series}{LNAI})}, \bibfield{editor}{\bibinfo{person}{Andreas
  Dengel}, \bibinfo{person}{Karsten Berns}, \bibinfo{person}{Thomas Breuel},
  \bibinfo{person}{Frank Bomarius}, {and} \bibinfo{person}{Thomas
  Roth{-}Berghofer}} (Eds.), Vol.~\bibinfo{volume}{5243}.
  \bibinfo{publisher}{Springer}, \bibinfo{pages}{324--331}.
\newblock


\bibitem[\protect\citeauthoryear{Schr{\"o}der and Pattinson}{Schr{\"o}der and
  Pattinson}{2009}]%
        {SchroderPattinson09a}
\bibfield{author}{\bibinfo{person}{Lutz Schr{\"o}der} {and}
  \bibinfo{person}{Dirk Pattinson}.} \bibinfo{year}{2009}\natexlab{}.
\newblock \showarticletitle{{PSPACE} Bounds for Rank-1 Modal Logics}.
\newblock \bibinfo{journal}{\emph{ACM Trans.\ Comput.\ Log.}}
  \bibinfo{volume}{10} (\bibinfo{year}{2009}), \bibinfo{pages}{13:1--13:33}.
\newblock


\bibitem[\protect\citeauthoryear{Schr{\"o}der and Pattinson}{Schr{\"o}der and
  Pattinson}{2011}]%
        {SchroderPattinson11MSCS}
\bibfield{author}{\bibinfo{person}{Lutz Schr{\"o}der} {and}
  \bibinfo{person}{Dirk Pattinson}.} \bibinfo{year}{2011}\natexlab{}.
\newblock \showarticletitle{Modular algorithms for heterogeneous modal logics
  via multi-sorted coalgebra}.
\newblock \bibinfo{journal}{\emph{Math.\ Struct.\ Comput.\ Sci.}}
  \bibinfo{volume}{21} (\bibinfo{year}{2011}), \bibinfo{pages}{235--266}.
\newblock


\bibitem[\protect\citeauthoryear{Schr{\"{o}}der, Pattinson, and
  Kupke}{Schr{\"{o}}der et~al\mbox{.}}{2009}]%
        {SchroderEA09}
\bibfield{author}{\bibinfo{person}{Lutz Schr{\"{o}}der}, \bibinfo{person}{Dirk
  Pattinson}, {and} \bibinfo{person}{Clemens Kupke}.}
  \bibinfo{year}{2009}\natexlab{}.
\newblock \showarticletitle{Nominals for Everyone}. In
  \bibinfo{booktitle}{\emph{Int.\ Joint Conf.\ Artificial Intelligence, {IJCAI}
  2009}}, \bibfield{editor}{\bibinfo{person}{Craig Boutilier}} (Ed.).
  \bibinfo{pages}{917--922}.
\newblock


\bibitem[\protect\citeauthoryear{Schr{\"{o}}der and Venema}{Schr{\"{o}}der and
  Venema}{2018}]%
        {SchroderVenema18}
\bibfield{author}{\bibinfo{person}{Lutz Schr{\"{o}}der} {and}
  \bibinfo{person}{Yde Venema}.} \bibinfo{year}{2018}\natexlab{}.
\newblock \showarticletitle{Completeness of Flat Coalgebraic Fixpoint Logics}.
\newblock \bibinfo{journal}{\emph{{ACM} Trans.\ Comput.\ Log.}}
  \bibinfo{volume}{19}, \bibinfo{number}{1} (\bibinfo{year}{2018}),
  \bibinfo{pages}{4:1--4:34}.
\newblock


\bibitem[\protect\citeauthoryear{Seidl, Schwentick, and Muscholl}{Seidl
  et~al\mbox{.}}{2008}]%
        {SeidlEA08}
\bibfield{author}{\bibinfo{person}{Helmut Seidl}, \bibinfo{person}{Thomas
  Schwentick}, {and} \bibinfo{person}{Anca Muscholl}.}
  \bibinfo{year}{2008}\natexlab{}.
\newblock \showarticletitle{Counting in trees}. In
  \bibinfo{booktitle}{\emph{Logic and Automata: History and Perspectives [in
  Honor of Wolfgang Thomas]}}, \bibfield{editor}{\bibinfo{person}{J{\"{o}}rg
  Flum}, \bibinfo{person}{Erich Gr{\"{a}}del}, {and} \bibinfo{person}{Thomas
  Wilke}} (Eds.). \bibinfo{publisher}{Amsterdam Univ.\ Press},
  \bibinfo{pages}{575--612}.
\newblock
\showISBNx{978-90-5356-576-6}


\bibitem[\protect\citeauthoryear{Tobies}{Tobies}{2001}]%
        {TobiesThesis}
\bibfield{author}{\bibinfo{person}{Stephan Tobies}.}
  \bibinfo{year}{2001}\natexlab{}.
\newblock \emph{\bibinfo{title}{Complexity results and practical algorithms for
  logics in Knowledge Representation}}.
\newblock \bibinfo{thesistype}{Ph.D. Dissertation}. \bibinfo{school}{RWTH
  Aachen}.
\newblock


\end{thebibliography}

\end{document}